\documentclass[11pt]{article}

\usepackage{amsfonts}
\usepackage{amsmath}
\usepackage{amssymb}
\usepackage{graphicx}
\usepackage{epstopdf}
\usepackage{float}
\usepackage{natbib}
\usepackage{theorem}

\renewcommand{\baselinestretch}{1.8}

\newtheorem{theorem}{Theorem}

\newtheorem{lemma}{Lemma}

\newtheorem{proposition}{Proposition}
\newcommand{\Perp}{\perp\!\!\!\perp}
\newcommand{\qed}{\hfill $\blacksquare$}
\newenvironment{proof}[1][Proof]{\begin{trivlist}
\item[\hskip \labelsep {\bfseries #1}]}{\end{trivlist}}

\begin{document}

\title{Pre-averaging estimators of the ex-post covariance matrix\\[-0.50cm]
in noisy diffusion models with non-synchronous data\thanks{Christensen and Podolskij were supported by CREATES, which is funded by the Danish National Research Foundation. Kinnebrock received financial support from the Rhodes Trust and Oxford-Man Institute of Quantitative Finance. We would like to thank the editor Ron Gallant, two anonymous referees, Neil Shephard and Mathias Vetter for valuable comments on previous drafts of this paper.}}
\author{Kim Christensen\thanks{CREATES, School of Economics and Management, Aarhus University, Building 1322, Bartholins All\'{e} 10, 8000 Aarhus, Denmark. E-mail \texttt{kchristensen@creates.au.dk}. Phone: +45 8942 1550. Corresponding author.} \and Silja Kinnebrock\thanks{Oxford-Man Institute of Quantitative Finance, University of Oxford, Eagle House, Walton Well Road, Oxford OX2 6ED, United Kingdom. E-mail: \texttt{silja.kinnebrock@merton.oxon.org}. } \and Mark Podolskij\thanks{ETH Z\"{u}rich, Department of Mathematics, R\"{a}mistrasse 101, CH-8092 Z\"{u}rich, Switzerland and affiliated with CREATES, Aarhus University, Denmark. E-mail: \texttt{mark.podolskij@math.ethz.ch}.}}
\date{April, 2010}
\maketitle

\vspace*{-1.0cm}

\begin{abstract}
We show how pre-averaging can be applied to the problem of measuring the ex-post covariance of financial asset returns under microstructure noise and non-synchronous trading. A pre-averaged realised covariance is proposed, and we present an asymptotic theory for this new estimator, which can be configured to possess an optimal convergence rate or to ensure positive semi-definite covariance matrix estimates. We also derive a noise-robust Hayashi-Yoshida estimator that can be implemented on the original data without prior alignment of prices. We uncover the finite sample properties of our estimators with simulations and illustrate their practical use on high-frequency equity data.



\end{abstract}

\vfill

\thispagestyle{empty}

\pagebreak

\renewcommand{\baselinestretch}{2.0}

\section{Introduction}
\setcounter{page}{1}

The theory of financial economics is often cast in multivariate settings, where the covariance structure of assets plays a key role to the solution of fundamental economic problems, such as optimal asset allocation and risk management. In recent years, a broader access to financial high-frequency data has improved our ability to accurately estimate and draw inference about financial covariation. The underlying idea is to use quadratic covariation, which we can estimate using realised covariance, as an ex-post measure, whose increments can be studied to learn about the properties of the true asset return covariation \citep*[e.g.][]{andersen-bollerslev-diebold-labys:03a,barndorff-nielsen-shephard:04a}.

In practice, implementing realised covariance is hampered by two empirical phenomena, namely the presence of market microstructure noise (e.g., price discreteness or bid-ask spread bounce) and non-synchronous trading. The impact of microstructure noise has received much attention in the univariate setting, where its effect on the realised variance has been well-documented. This builds on previous work in the noiseless case, including \citet*{andersen-bollerslev-diebold-labys:01a}, \citet*{barndorff-nielsen-shephard:02a}, or \citet*{mykland-zhang:06a,mykland-zhang:09a}. A key to understanding the nature of the noise and a possible tool of how to deal with it is that microstructure noise induces autocorrelation in high-frequency returns and this leads to a bias problem \citep*[see, e.g.,][]{zhou:96a,ait-sahalia-mykland-zhang:05a,hansen-lunde:06a}. Currently, there are three main univariate approaches, where the damage caused by the noise is explicitly fixed: the two-scale subsampler proposed by \citet*{zhang-mykland-ait-sahalia:05a} or its multi-scale version of \citet*{zhang:06a}, the realised kernel introduced in \citet*{barndorff-nielsen-hansen-lunde-shephard:08a}, which relies on autocovariance-based corrections, and finally the pre-averaging estimator of \citet*{podolskij-vetter:09a} and \citet*{jacod-li-mykland-podolskij-vetter:09a}.

The multivariate version of this problem is, however, more complicated in that not only does the estimator need to be robust against various types of noise, it also has to cope with non-synchronous trading \citep*[see, e.g.,][]{fisher:66a}. Asynchronicity causes high-frequency covariance estimates to be biased towards zero as the sampling frequency increases. This feature of the data, the so-called Epps effect, was highlighted by \citet*{epps:79a}. Intuitively, as the sampling frequency is increased, there are more and more zero-returns in the presence of non-synchronous trading, and this will dominate realised covariance and related statistics (e.g. realised correlation). \citet*{hayashi-yoshida:05a} introduced an estimator, which is capable of dealing with non-synchronous data, but not with market microstructure noise. More recently, \citet*{zhang:08a} extended the two-scale RV to integrated covariance estimation in the simultaneous presence of noise and non-synchronicity, while in concurrent and independent work \citet*{barndorff-nielsen-hansen-lunde-shephard:08c} proposed a multivariate realised kernel. Additional work in this growing line of research includes \citet*{malliavin-mancino:02a}, \citet*{martens:03a}, \citet*{reno:03a}, \citet*{bandi-russell:05a}, \citet*{griffin-oomen:06a}, \citet*{large:07a}, \citet*{voev-lunde:07a}, and \citet*{boudt-croux-laurent:08a}, among others.

In this paper, we propose to use a "modulated" realised covariance (MRC) to estimate the ex-post integrated covariance. The econometric technique employed here for dealing with microstructure noise relies on rather simple pre-averaging of the high-frequency data, which makes the estimator both intuitive to understand and trivial to implement. It relates to previous work in the univariate case, where pre-averaging has been suggested in \citet*{podolskij-vetter:09a} and \citet*{jacod-li-mykland-podolskij-vetter:09a}. The current article draws ideas from these papers, but the multivariate extension is challenging, as it faces the additional complexity of non-synchronous trading and requires that the resulting estimator be positive semi-definite.

The pre-averaging approach depends on a bandwidth parameter, or window length, that grows with the sample and dictates the amount of averaging to be carried out. In turn, the choice of this tuning parameter controls the influence of microstructure noise on the MRC and, hence, also its asymptotic properties. In the optimal case, called balanced pre-averaging, this leads to an efficient $n^{- 1 / 4}$ rate of convergence, which is known to be the fastest attainable \citep*[see, e.g.,][]{gloter-jacod:01a,gloter-jacod:01b}. This baseline MRC estimator, however, needs a bias-correction to be consistent for the integrated covariance. As a result, it is not guaranteed to be positive semi-definite in finite samples, though our empirical work indicates this shortcoming is not too much of a concern for more recent data. Nonetheless, as we show in the paper, it is straightforward to design a positive semi-definite estimator by increasing the pre-averaging window length slightly, which can also serve to make the MRC robust against more general noise processes.

The MRC is, in all its essence, a realised covariance computed on the back of pre-averaged high-frequency returns. As such, it depends on receiving synchronous observations as input, which clashes with the irregular spacing of real high-frequency data. We propose two distinct ways in which pre-averaging can be applied in the context of non-synchronous trading. First, we use traditional imputation schemes to map asynchronous data onto a common time grid, for example using previous-tick or refresh time, where the latter approach has been used in \citet*{barndorff-nielsen-hansen-lunde-shephard:08c}. An MRC computed from such returns will be asymptotically robust to non-synchronous trading. Second, we extend the \citet*{hayashi-yoshida:05a} estimator to the case of microstructure noise by using pre-averaging and show that it is consistent. This second estimator has the property that it can be implemented directly on the irregular, non-synchronous and noisy observations without any form of imputation. It therefore omits throwing away information in the sample and further avoids potential biases arising from artificially imputed returns.

An appealing feature of pre-averaging is that it is a general statistical tool that can be applied to many estimation problems. This proves useful in our setting, because as usual the mixed Gaussian central limit theorems feature an unknown conditional covariance matrix. In practice, this must be robustly estimated from sample data in the presence of noise and non-synchronous trading to make the distributional results feasible, such that confidence bands for elements of the integrated covariance matrix can be constructed. We outline how this can be done based on pre-averaged high-frequency data.

The paper progresses as follows. In section 2, we formulate the theoretical setup and define the MRC estimator. In section 3, we first show consistency of the MRC based on balanced pre-averaging and then derive its asymptotic distribution. As discussed above, this estimator needs a bias-correction, so we carry on to study a modified MRC estimator, in which the degree of pre-averaging is increased. We also discuss the application of the MRC to non-synchronous data, show its relation to the multivariate realised kernel, and finally we derive a pre-averaged version of the Hayashi-Yoshida estimator. In section 4, we propose an estimator of the conditional covariance matrix that appears in the central limit theorem of MRC, which can be used to transform infeasible limit results into feasible ones. In section 5, the focus is shifted towards regression and correlation analysis. A simulation study is undertaken in section 6 to uncover the finite sample properties of our estimators, while an empirical illustration is conducted in section 7. Section 8 draws conclusions and presents some ideas for future work. The appendix contains the derivations of all our theoretical results.

\section{Theoretical setup}
We consider a vector of log-prices $X$ defined on a probability space $(\Omega^0, \mathcal F^0, P^0)$ and equipped with an information filtration $(\mathcal F_t^0)_{t \geq 0}$. $X$ has dimension $d$ - the number of assets under consideration.

A standard no-arbitrage condition suggests security prices must be semimartingales \citep*[see, e.g.,][]{back:91a,delbaen-schachermayer:94a}. These processes obey the fundamental theorem of asset pricing and, as a result, are used extensively to model the evolution of asset prices through time. In accordance with this, we model $X$ as a semimartingale that follows the equation
\begin{equation}
\label{bsm}
X_{t} = X_{0} + \int_{0}^{t} a_{u} \text{d}u  + \int_{0}^{t} \sigma_{u} \text{d}W_{u}, \quad t \geq 0,
\end{equation}
where $a = \left( a_{t} \right)_{t \geq 0}$ is a $d$-dimensional predictable locally bounded drift vector, $\sigma = \left( \sigma_{t} \right)_{t \geq 0}$ an adapted c\`{a}dl\`{a}g $d \times d$ covolatility matrix and $W = \left( W_{t} \right)_{t \geq 0}$ is $d$-dimensional Brownian motion.

This model is a Brownian semimartingale, or stochastic volatility model with drift, which permeates financial economics \citep*[cf.,][for a review]{ghysels-harvey-renault:96a}. We think of this construct as governing an underlying efficient price process - the price that would prevail in the absence of market frictions, which we then subject to microstructure noise.

Of importance to our analysis is the quadratic covariation process of $X$, which is defined as
\begin{equation}
\label{qv}
\left[ X \right]_{t} = \underset{n \to \infty}{\text{p-lim}} \sum_{i = 1}^{n} \left( X_{t_{i}} - X_{t_{i - 1}} \right) \left( X_{t_{i}}  - X_{t_{i - 1}} \right)'
\end{equation}
for any sequence of deterministic partitions $0 = t_0  < t_1  < ... < t_n  = t$ with $\sup_{i} \left\{ t_{i} - t_{i - 1} \right\} \to 0$ for $n \to \infty$. In our setting, the quadratic covariation of $X$ is given by
\begin{equation}
\label{qvbsm}
\left[ X \right]_{t} = \int_{0}^{t} \Sigma_{u} \text{d}u,
\end{equation}
where $\Sigma = \sigma \sigma'$. The quadratic covariation is pivotal in financial economics \citep*[see, e.g., the reviews by][]{barndorff-nielsen-shephard:07a,andersen-bollerslev-diebold:09a}, and we thus take Eq. \eqref{qvbsm} as defining the target that we are interested in estimating. We note that for obvious reasons the matrix in Eq. \eqref{qvbsm} is also called the integrated covariance and both terms are used interchangeably.

Throughout the remainder of the paper, and without loss of generality, we restrict the clock $t$ to evolve in the unit interval $[0,1]$, which we think of as representing the passing of an economic event, for example a trading day.

\subsection{Microstructure noise}
In practice, market microstructure noise leads to a departure from the pure semimartingale model. Microstructure noise has many forms, including price discreteness and bid-ask spread bounce, which creates spurious variation in asset prices. As a result, we do not observe $X$ from Eq. \eqref{bsm} in the market but a process $Y$, which is the efficient price distorted by noise. More precisely, we consider the process $Y$, observed at time points $i / n$, $i = 0, 1, \ldots, n$, which is given as
\begin{equation}
\label{zsh}
Y_{t} = X_{t} + \epsilon_{t},
\end{equation}
where $(\epsilon_{t})$ is an i.i.d. process with $X \Perp \epsilon$ (the symbol $\Perp$ is used to denote stochastic independence).

The noise process can be constructed as follows. We define a second probability space $(\Omega^1, \mathcal F^{1}, ( \mathcal{F}_{t}^{1})_{t \geq 0}, P^1)$, where $\Omega^{1}$ denotes $\mathbb{R}^{[0,1]}$ and $\mathcal{F}^{1}$ the product Borel-$\sigma$-field on $\Omega^{1}$. Next, let $Q$ be a probability measure on $\mathbb{R}$ ($Q$ is the marginal distribution of $\epsilon$). For any $t\in [0,1]$, $P_t^{1}=Q$ and $P^1$ denotes the product $\otimes_{t \in [0,1]} P_t^{1}$. The filtered probability space $(\Omega, \mathcal F, (\mathcal{F}_{t})_{t \geq 0}, P)$, on which we define the process $Y$, is given as
\begin{equation}
\left.
\begin{array}{l}
\Omega = \Omega^{0} \times \Omega^{1}, \qquad \mathcal{F} = \mathcal{F}^{0} \times \mathcal{F}^{1}, \qquad \mathcal{F}_{t} = \bigcap_{s>t} \mathcal{F}_s^{0} \times \mathcal{F}_s^{1}, \\
P = P^{0} \otimes P^{1}.
\end{array}
\right\}
\end{equation}
The multivariate noise process $\epsilon$ is assumed to satisfy:
\begin{equation}
\label{assnoise2} \mathbb{E} \left( \epsilon_t \right) = 0, \qquad \mathbb{E} \left( \epsilon_t \epsilon_t' \right) = \Psi,
\end{equation}
where $\Psi$ is a positive definite $d \times d$-matrix.

\noindent \textbf{Remark} \ The empirical results found by \citet*{hansen-lunde:06a} show that both the i.i.d. assumption on $(\epsilon_t)$ and the independence $X \Perp \epsilon$ can be called into question when sampling the data at very high frequencies, e.g., below the 1-minute mark \citep*[see also][]{diebold-strasser:08a}. \citet*{jacod-li-mykland-podolskij-vetter:09a} consider more general types of (1-dimensional) noise processes. Roughly speaking, they assume that the errors $\epsilon_t$'s are, conditionally on $X$, centered and independent. The asymptotic theory developed in this paper still holds true for the multivariate version of such noise processes, but we restrict attention to models of the form in Eq. \eqref{zsh} to ease the exposition.

\subsection{Pre-averaging of high-frequency data}
\label{preavg} It is intuitive that under mean zero i.i.d. microstructure noise some form of smoothing of the observed log-price $Y$ should tend to diminish the impact of the noise. Effectively, we are going to approximate $X_{t}$, $X$ being a continuous function of $t$, by an average of observations of $Y$ in a neighborhood of $t$, the noise being averaged away.

Here, we describe in more detail how to conduct the pre-averaging. In particular, we consider a sequence of integers, $k_{n}$, and a number $\theta \in (0, \infty)$ such that
\begin{equation}
\label{kthetarelation}
\frac{k_{n}}{\sqrt{n}} = \theta + o \bigl( n^{- 1 / 4} \bigr).
\end{equation}
An example of this would be $k_{n} = \left\lfloor \theta \sqrt{n} \right \rfloor$.

We also choose a function $g$ on $[0,1]$, which is continuous, piecewise continuously differentiable with a piecewise Lipschitz derivative $g'$ with $g(0) = g(1) = 0$ and which satisfies $\int_0^1 {g^2 \left( s \right) \text{d}s > 0}$. Furthermore, we introduce the following functions and numbers that are associated with $g$:
\begin{equation*}
\phi_1\left( s \right) = \int_s^1 {g'\left( u \right)g'\left( {u - s} \right)\text{d}u},
\qquad
\phi_2\left( s \right) = \int_s^1 {g\left( u \right)g\left( {u - s} \right)\text{d}u}, \qquad
\psi _1  = \phi_1 \left( 0 \right),
\qquad
\psi _2  = \phi_2 \left( 0 \right),
\end{equation*}
\begin{equation*}
\Phi _{11} = \int_0^1 {\phi_1^2 \left( s \right)\text{d}s},
\qquad
\Phi _{12} = \int_0^1 {\phi_1 \left( s \right)\phi_2 \left( s \right)\text{d}s},
\qquad
\Phi _{22} = \int_0^1 {\phi_2^2 \left( s \right)\text{d}s}.
\end{equation*}
The functions $\phi_1$ and $\phi_2$ are assumed to be $0$ outside the interval $[0,1]$.

Next, with any process $V = (V_t)_{t \ge 0}$ we associate the following random variables
\begin{equation*}
\Delta _i^n V = V_{ \frac{i}{n}}  - V_{ \frac{i - 1}{n}}, \ \text{for } i = 1, \ldots, n
\qquad
\bar{V}_{i}^{n} = \sum_{j = 1}^{k_{n} - 1} g \left( \frac{j}{k_{n}} \right) \Delta_{i + j}^{n} V, \ \text{for } i = 0, \ldots, n - k_{n} + 1.
\end{equation*}
Applying this notation to $Y$, it can be seen that $\Delta_{i}^{n} Y$ represents the noisy high-frequency returns, while $\bar{Y}_{i}^{n}$ is the pre-averaged return data, using the weight function $g$. It follows that the stochastic order of $\bar{Y}_{i}^{n} = \bar{X}_{i}^{n} + \bar{ \epsilon}_{i}^{n}$ is controlled by the sequence $k_{n}$, since
\begin{equation}
\label{Eqn:uorder}
\bar{X}_{i}^{n} = O_p \left( \sqrt{ \frac{k_{n}}{n}} \right), \qquad
\bar{ \epsilon}_{i}^{n} = O_p \left( \sqrt{ \frac{1}{k_{n}}} \right).
\end{equation}
Thus, taking $k_{n} = O( \sqrt{n})$ implies that the orders of the two terms in Eq. \eqref{Eqn:uorder} are equal, so that $\bar{Y}_{i}^{n} = O_{p} \left( n^{-1/4} \right)$. This is called balanced pre-averaging and delivers the best rate of convergence. As shown below, it is also useful to look at cases in which a higher order of $k_{n}$ is chosen. This results in a suboptimal rate of convergence, but it has some potentially valuable side-effects on the robustness and finite sample properties of our estimator.

The pre-averaging window length, $k_{n}$, depends on the tuning parameter $\theta$, which needs to be chosen by the user. We will later discuss how to sensibly make this choice.

\subsection{Modulated realised covariance}
The core statistic of this paper is the multivariate extension of the estimator, which was introduced in \citet*{jacod-li-mykland-podolskij-vetter:09a}. We call it the modulated realised covariance (MRC) and define it as
\begin{equation}
\label{MMRC}
MRC \left[ Y \right]_{n} = \frac{n}{n - k_{n} + 2} \frac{1}{ \psi_{2} k_{n}}  \sum_{i = 0}^{n - k_{n} + 1} \bar{Y}_{i}^{n} \left( \bar{Y}_{i}^{n} \right)'.
\end{equation}
The factor $n / (n - k_{n} + 2)$ is a finite sample correction for the true number of summands in $MRC \left[ Y \right]_{n}$ relative to the sample size $n$. It is sometimes left out in the presentation below, but it is always included in implementations on data.

\noindent \textbf{Remark} \ The sum of outer products in Eq. \eqref{MMRC} is a realised covariance based on pre-averaged data. To build some intuition for our approach, we explain the usage of $\bar Y_i^{n}$ in more detail. Suppose $k_{n}$ is an even number and write
\begin{equation*}
\hat Y_{i}^{n} = \frac{2}{k_{n}} \sum_{j = 0}^{k_{n} / 2 - 1} Y_{ \frac{i + j}{n}},
\end{equation*}
which is a simple average of $Y$ over $k_{n} / 2$ terms. Because of this pre-averaging, $\hat Y_i^{n}$ will be closer to the efficient price $X_{ \frac{i}{n}}$. Next, we compute the realised covariation estimator based on these filtered increments by setting
\begin{equation*}
\bar Y_i^{n} = \frac{1}{2} (\hat Y_{i + \frac{k_{n}}{2}}^{n} - \hat Y_{i}^{n}) = \frac{1}{k_{n}} \left( \sum_{j = k_{n} / 2}^{k_{n} - 1} Y_{ \frac{i + j}{n}} - \sum_{j = 0}^{k_{n} / 2 - 1} Y_{ \frac{i + j}{n}} \right).
\end{equation*}
(However, as we shall see this induces a bias, which is a function of $\Psi$). This method was originally proposed by \citet*{podolskij-vetter:09a} and using the above definition of $\bar Y_i^{n}$ corresponds to choosing the weight function
\begin{equation}
g \left( x \right) = \min \left( x, 1 - x \right),
\end{equation}
which is the most intuitive example. Here we explicitly give the numerical values of the asymptotic constants for this choice of function $g$, as it is the one used for all our simulations and empirical work:
\begin{equation*}
\psi _1  = 1,\qquad \psi _2  = \frac{1}{{12}},\qquad \Phi _{11}  = \frac{1}{6},\qquad \Phi _{12}  = \frac{1}{{96}},\qquad \Phi _{22}  = \frac{{151}}{{80640}}.
\end{equation*}

\noindent \textbf{Remark} \ As noted in \citet*{jacod-li-mykland-podolskij-vetter:09a}, to avoid biases in small samples we should replace the asymptotic constants and functions $\psi_{1}$, $\psi_{2}$, $\phi_{1}$, $\phi_{2}$, $\Phi_{11}$, $\Phi _{12}$, and $\Phi _{22}$ by their Riemann approximations:
\begin{equation*}
\psi_{1}^{k_{n}} = k_{n} \sum_{i = 1}^{k_{n}} \left( g \left( \frac{i}{k_{n}} \right) - g \left( \frac{i - 1}{k_{n}} \right) \right)^{2}, \quad \psi_{2}^{k_{n}} = \frac{1}{k_{n}} \sum_{i = 1}^{k_{n} - 1} g^{2} \left( \frac{i}{k_{n}} \right),
\end{equation*}
\begin{equation*}
\phi_{1}^{k_{n}} \left( j \right) = \sum_{i = j + 1}^{k_{n} - 1} \left( g \left( \frac{i - 1}{k_{n}} \right) - g \left( \frac{i}{k_{n}} \right) \right) \left( g \left( \frac{i - j - 1}{k_{n}} \right) - g \left( \frac{i - j}{k_{n}} \right) \right), \quad
\phi_{2}^{k_{n}} \left( j \right) = \sum_{i = j + 1}^{k_{n} - 1} g \left( \frac{i}{k_{n}} \right) - g \left( \frac{i - j}{k_{n}} \right),
\end{equation*}
\begin{equation*}
\Phi_{11}^{k_{n}} = k_{n} \left( \sum_{j = 0}^{k_{n} - 1} \left( \phi_{1}^{k_{n}} \left( j \right) \right)^{2} - \frac{1}{2} \left( \phi_{1}^{k_{n}} \left( 0 \right) \right)^{2} \right), \quad \Phi_{12}^{k_{n}} = \frac{1}{k_{n}} \left( \sum_{j = 0}^{k_{n} - 1} \phi_{1}^{k_{n}} \left( j \right) \phi_{2}^{k_{n}} \left( j \right) - \frac{1}{2} \phi_{1}^{k_{n}} \left( 0 \right) \phi_{2}^{k_{n}} \left( 0 \right) \right),
\end{equation*}
\begin{equation*}
\Phi_{22}^{k_{n}} = \frac{1}{k_{n}^{3}} \left( \sum_{j = 0}^{k_{n} - 1} \left( \phi_{2}^{k_{n}} \left( j \right) \right)^{2} - \frac{1}{2} \left( \phi_{2}^{k_{n}} \left( 0 \right) \right)^{2} \right).
\end{equation*}
These are the actual terms that appear in the computations of the MRC. Note that for all appropriate indices of $i$ and $j$, $\psi_{i}^{k_{n}} \to \psi_{i}$, $\phi_{i}^{k_{n}} \to \phi_{i}$, $\Phi_{ij}^{k_{n}} \to \Phi_{ij}$ as $n \to \infty$ at smaller order than $n^{-1 / 4}$, so the finite sample versions can be replaced with the asymptotic constants in all the limit theorems given below, including the central limit theorems.

\section{Asymptotic properties of MRC}
\subsection{Consistency}
Our first result inspects the probability limit of $MRC \left[ Y \right]_{n}$.
\begin{theorem}
\label{consistency}
Assume that $\mathbb{E} \left( | \epsilon^{j} |^{4} \right) < \infty$ for all $j = 1,...,d$ and $(k_{n}, \theta)$ satisfy Eq. \eqref{kthetarelation}. As $n \to \infty$, it holds that
\begin{equation}
\label{stochcon2}
MRC \left[ Y \right]_{n} \overset{p}{ \to} \int_{0}^{1} \Sigma_{s} \text{\upshape{d}}s + \frac{ \psi_{1}}{ \theta^{2} \psi_{2}} \Psi.
\end{equation}
\end{theorem}
\begin{proof}
See appendix. \qed
\end{proof}
A couple of points are worth highlighting. First, as Theorem \ref{consistency} shows $MRC \left[ Y \right]_{n}$ is consistent for $\int_{0}^{1} \Sigma_{s} \text{d}s$ up to a bias-correction. The bias term depends on the unknown $\Psi$, which must be estimated from the data.

We set
\begin{equation}
\hat{ \Psi}_{n} = \frac{1}{2n} \sum_{i = 1}^{n} \Delta_{i}^{n} Y \left( \Delta_{i}^{n} Y \right)',
\end{equation}
which is linked to the univariate estimator proposed by \citet*{ait-sahalia-mykland-zhang:05a} and \citet*{bandi-russell:06a,bandi-russell:08a}. Then, we obtain the convergence $\hat{ \Psi}_{n} \overset{p}{ \to} \Psi$, such that
\begin{equation*}
MRC \left[ Y \right]_{n} - \frac{ \psi_{1}^{k_{n}}}{ \theta^{2} \psi_{2}^{k_{n}}} \hat{ \Psi}_{n} \overset{p}{ \to} \int_{0}^{1} \Sigma_{s} \text{\upshape{d}}s.
\end{equation*}
Hence, for the remainder of the paper, we shall incorporate the bias-correction term into the definition of $MRC \left[ Y \right]_{n}$.\footnote{Since $\sum_{i = 1}^{n} \Delta_{i}^{n} Y \left( \Delta_{i}^{n} Y \right)' = 2 n \Psi + \int_{0}^{1} \Sigma_{s} \text{d}s + o_{p}(n^{-1})$, where the error of this approximation has expectation zero, the bias-corrected $MRC \left[ Y \right]_{n}$ actually estimates $\left(1 - \frac{ \psi_{1}^{k_{n}}}{ \theta^{2} \psi_{2}^{k_{n}}} \frac{1}{2n} \right) \int_{0}^{1} \Sigma_{s} \text{d}s$ and thus needs to be rescaled by $1 / \left(1 - \frac{ \psi_{1}^{k_{n}}}{ \theta^{2} \psi_{2}^{k_{n}}} \frac{1}{2n} \right)$. We include this rescaling in our simulations and empirical work but omit it throughout the remainder of the text to simplify notation.} In doing so, we are no longer ensured that $MRC \left[ Y \right]_{n}$ is positive semi-definite in finite samples. To deal with this problem, we propose an alternative formulation of the MRC below, which uses a longer pre-averaging window $k_{n}$ to avoid this step. Second, since $\hat{ \Psi}_{n}$ is a $\sqrt{n}$-estimator of $\Psi$, it will not affect the CLT of $MRC \left[ Y \right]_{n}$, as the latter converges at a slower rate. Third, our initial MRC estimator is based on synchronous data, assuming all components of the log-price vector $Y$ are observed contemporaneously. We will later extend the MRC to the non-synchronous setting.

\subsection{The central limit theorem}
We proceed with the central limit theorem for $MRC \left[ Y \right]_{n}$. As in \citet*{jacod-li-mykland-podolskij-vetter:09a}, we only
require a moment condition on $\epsilon$ to prove this result.

The notion of stable convergence is used, which we describe here. A sequence of random variables $V^{n}$ is said to converge stably in law towards $V$, where $V$ is defined on an appropriate extension $(\Omega', \mathcal{F}', P')$ of the probability space $(\Omega, \mathcal{F},P)$, if and only if for any $\mathcal{F}$-measurable, bounded random variable $W$ and any bounded, continuous function $f$, the convergence
\begin{equation*}
\lim_{n \to \infty} \mathbb{E} \left[ W f(V^{n}) \right] = \mathbb{E}' \left[ W f(V) \right]
\end{equation*}
holds.

We write this as $V^{n} \overset{d_{s}}{ \to} V$ and note that stable convergence is a slightly stronger mode of convergence than weak convergence, or convergence in law, which is the special case obtained by taking $W = 1$ \citep*[see, e.g.,][for further details on stable convergence]{renyi:63a,aldous-eagleson:78a}. \citet*{jacod-shiryaev:03a} discuss the extension of this concept to stable convergence of processes. The key reason we require the convergence in law stably is that the conditional covariance matrix in the CLT of $MRC[Y]_{n}$, $\text{avar}_{\text{MRC}_{n}}$, is a function of $\sigma$ and therefore random, and the usual convergence in law is insufficient to ensure joint convergence of the bivariate vector $(MRC[Y]_{n}, \text{avar}_{\text{MRC}_{n}})$, which we need to apply the delta method to the joint asymptotic distribution and to construct confidence intervals.

\begin{theorem}
\label{clt2}
Assume that $\mathbb{E} \left( | \epsilon^{j} |^{8} \right) < \infty$ for all $j = 1,...,d$ and $(k_n, \theta)$ satisfy Eq. \eqref{kthetarelation}. As $n \to \infty$, it holds that
\begin{equation}
\label{mainstats}
n^{1 / 4} \left( MRC \left[ Y \right]_{n} - \int_{0}^{1} \Sigma_{s} \text{\upshape{d}}s \right) \overset{d_{s}}{ \to} \sum_{j', k' = 1}^{d} \int_{0}^{1} \gamma_{s}^{jk,j'k'} \text{\upshape{d}} B_{s}^{j'k'},
\end{equation}
where $B$ is a standard $d^2$-dimensional Brownian motion defined on an extension of $(\Omega, \mathcal F, (\mathcal F_{t})_{t \geq 0}, P)$ with $B \Perp \mathcal{F}$,
\begin{equation*}
\sum_{j,m = 1}^{d} \gamma_{s}^{kl, jm} \gamma_{s}^{k'l', jm}
= \frac{2}{ \psi_{2}^{2}} \left( \Phi_{22} \theta \Lambda_{s}^{kl, k'l'} + \frac{ \Phi_{12}}{ \theta} \Theta_{s}^{kl, k'l'} + \frac{ \Phi_{11}}{ \theta^{3}} \Upsilon^{kl, k'l'} \right),
\end{equation*}
and where $\Lambda$, $\Theta$ and $\Upsilon$ are $d \times d \times d \times d$ arrays with elements
\begin{equation}
\label{ltu}
\begin{array}{l}
 \Lambda _s  = \left\{ {\Sigma _s^{kk'} \Sigma _s^{ll'}  + \Sigma _s^{kl'} \Sigma _s^{lk'} } \right\}_{k,k',l,l' = 1,...,d},  \\
  \\
 \Theta_s = \left\{ {\Sigma _s^{kk'} \Psi ^{ll'}  + \Sigma _s^{kl'} \Psi ^{k'l}  + \Sigma _s^{k'l} \Psi ^{kl'} + \Sigma _s^{ll'} \Psi ^{kk'} } \right\}_{k,k',l,l' = 1,...,d},  \\
  \\
 \Upsilon = \left\{ {\Psi ^{kk'} \Psi ^{ll'}  + \Psi ^{kl'} \Psi ^{lk'} } \right\}_{k,k',l,l' = 1,...,d}.  \\
 \end{array}
\end{equation}
\end{theorem}
\begin{proof}
See appendix. \qed
\end{proof}
Because $B \Perp \mathcal{F}$, we can write the convergence statement in Theorem \ref{clt2} as follows:
\begin{equation*}
n^{1/4} \left( MRC \left[ Y \right]_{n} - \int_{0}^{1} \Sigma_{s} \text{d}s \right) \overset{d_{s}}{ \to} MN \left( 0, \text{avar}_{\text{MRC}} \right),
\end{equation*}
where
\begin{equation}
\label{L2}
\text{avar}_{\text{MRC}} = \frac{2}{{\psi _2^2 }} \left( \Phi_{22} \theta \int_0^1 \Lambda_s \text{d}s + \frac{\Phi_{12}}{ \theta} \int_{0}^{1} \Theta_{s} \text{d}s  + \frac{{\Phi _{11} }}{{\theta ^3 }}\Upsilon \right)
\end{equation}
is the conditional covariance matrix. This means that the asymptotic distribution of $MRC \left[ Y \right]_{n}$ is mixed normal. To make use of this result to construct confidence intervals for elements of $\int_{0}^{1} \Sigma_{u} \text{d}u$ in practice, we need to estimate $\text{avar}_{\text{MRC}}$, which is addressed in section 4.\footnote{The assumption that data be equidistant is not required for the consistency to hold true. It would also apply under identical observation times (i.e., synchronous but non-equidistant data). Here, $MRC \left[ Y \right]_{n}$ is consistent with $\Delta_{i}^{n} V$ redefined as $\Delta_{i}^{n} V = V_{t_{i}} - V_{t_{i-1}}$ for any process $V$ and $k_{n} = \theta \sqrt{n} + o(n^{1 / 4})$. This is not surprising: the realised covariance also remains consistent for irregular observations (by definition). The CLT also holds, but here the variable of integration $\text{d}s$ needs to be replaced by $\text{d}H_{s}$, where $ H_{s}$ is the so-called "quadratic variation of time", see, e.g., \citet*{barndorff-nielsen-hansen-lunde-shephard:08c}.}

\subsection{Choosing $\theta$ in practice}
The $\text{avar}_{\text{MRC}}$ matrix in Theorem \ref{clt2} depends on the parameter $\theta$, or in other words the window size $k_{n}$. If the purpose is to estimate some one-dimensional parameters (such as realised covariance, regression or correlation) by real-valued functions of the MRC, it is in principle possible to minimize $\text{avar}_{\text{MRC}}$ by choosing the "best" $\theta$ for a fixed function $g$.\footnote{The optimal choice of $\theta$ will depend on the original real-valued functions of the MRC. In this sense, there is no universal optimal $\theta$.}

To illustrate this point, we focus on the estimation problem in the univariate case, $d = 1$. In this situation, the expressions reduce to
\begin{equation*}
\text{avar}_{\text{MRC}} = \frac{4}{ \psi_{2}^{2}} \left( \Phi_{22} \theta \int_0^1 \sigma_{s}^{4} \text{d}s + \frac{2 \Phi_{12}}{ \theta} \Psi^{2} \int_{0}^{1} \sigma_{s}^{2} \text{d}s  + \frac{{\Phi _{11} }}{{\theta ^3 }} \Psi^{4} \right),
\end{equation*}
where $IV = \int_{0}^{1} \sigma_{s}^{2} \text{d}s$ and $IQ = \int_{0}^{1} \sigma_{s}^{4} \text{d}s$ are called the integrated variance and integrated quarticity, respectively. Minimizing this term with respect to $\theta$ results in solving a quadratic equation. Thus, for the optimal choice of $\theta$, say $\theta^{*}$, we get $\theta^{*} = f \left(IV, IQ, \Psi \right)$. Then, we can use an iterative procedure to find an approximation of $\theta^{*}$: (i) Choose a "reasonable" value for $\theta$ and compute $\hat{IV}$, $\hat{IQ}$ and $\hat{ \Psi}$, (ii) from these compute $\hat{ \theta}^{*}$ by setting $\hat{ \theta}^{*} = f ( \hat{IV}, \hat{IQ}, \hat{ \Psi} )$ and (iii) then, assuming the values of $\hat{ \theta}^{*}$ converge, repeat this process until $\hat{ \theta}^{*}$ has stabilized.

In practice, a simple guide that informs us about how to select $\theta$ for small values of $n$ would certainly be valuable. Unfortunately, $\theta$ comes from asymptotic statistics and therefore it does not give any precise instructions on this issue. Nonetheless, some plausible range of values of $\theta$ can be inferred from previous work in this area. For example, \citet*{jacod-li-mykland-podolskij-vetter:09a} report that the univariate pre-averaged realised variance measure is fairly robust to the choice of $k_{n}$, and they suggest to take $\theta = 1/3$. \citet*{christensen-oomen-podolskij:10a} use simulations to gauge the influence of sample size and noise on the optimal choice of $\theta$. Interestingly, they show that the MSE curve of their pre-averaged quantile-based realised variance is highly asymmetric in $k_n$ and they generally prefer to use a slightly higher value of $k_{n}$ than what would be optimal. In their empirical analysis, they also show that conservative choices of $k_{n}$ helps to heavily reduce the detrimental effects of price discreteness. In the simulation section and empirical illustration below, we therefore pick conservative values of $k_n$ by choosing $\theta = 1$, which seems to work well.

\subsection{Positive semi-definite estimators}
In the previous section, we used an optimal pre-averaging window length to construct the MRC, which balances the impact of the noise with the estimation of the integrated covariance matrix. This choice leads to an optimal rate of convergence - $n^{- 1 / 4}$ - but requires that we subtract an estimate of $\Psi$ to eliminate the bias induced by noise, and the final estimator is then not positive semi-definite in general. Here, we demonstrate how positive semi-definite estimates of $\int_{0}^{1} \Sigma_{s} \text{d}s$ can be formed by increasing the bandwidth parameter $k_n$ as to kill the influence of the noise, rather than balancing it. This comes at the cost of slowing down the speed at which MRC converges to the true integrated covariation.

Now, we take:
\begin{equation} \label{deltanon}
\frac{k_{n}}{n^{1 / 2 + \delta}} = \theta + o \left( n^{- 1 / 4 + \delta / 2} \right)
\end{equation}
for some $0<\delta < 1/2$, and set
\begin{equation}
\label{MMRCpos}
MRC \left[ Y \right]_{n}^{ \delta} = \frac{n}{n - k_{n} + 2} \frac{1}{ \psi_{2} k_{n}} \sum_{i = 0}^{n - k_{n} + 1} \bar{Y}_{i}^{n} \left( \bar{Y}_{i}^{n}  \right)'.
\end{equation}
The following result shows that $MRC \left[ Y \right]_{n}^{ \delta}$ is consistent without a bias-correction.

\begin{theorem}
\label{stochnon}
Assume that $\mathbb{E} \left( | \epsilon^{j} |^{4} \right) < \infty$ for all $j = 1,...,d$ and $(k_n, \theta)$ satisfy Eq. \eqref{deltanon}. As $n \to \infty$, it holds that
\begin{equation}
\label{stochconnon}
MRC \left[ Y \right]_n^\delta \overset{p}{ \to} \int_{0}^{1} \Sigma_{s} \text{\upshape{d}}s.
\end{equation}
\end{theorem}
\begin{proof}
See appendix. \qed
\end{proof}
In Theorem \ref{stochnon}, the properties of the noise process do not show up in the stochastic limit of Eq. \eqref{stochconnon}, because the influence of the noise is negligible by the choice of order for $k_{n}$ made in Eq. \eqref{deltanon} (refer back to Eq. \eqref{Eqn:uorder}). This has some appealing advantages. First, $MRC \left[ Y \right]_{n}^{ \delta}$ is positive semi-definite by construction. Second, although we state and prove this result in the i.i.d. noise case, Theorem \ref{stochnon} does in fact allow for more general noise dynamics than in Theorem \ref{consistency}. In particular, so long as $\mathbb{E} \left( \epsilon_{t} \mid X \right) = 0$ and $\bar{ \epsilon}_{i}^{n}$ admits asymptotic normality at the usual rate $k_{n}^{- 1 / 2}$, the theorem will hold (so we do not require any assumptions on the dependence between $X$ and $\epsilon$). Of course, the rate $k_{n}^{- 1 / 2}$ is achieved in the i.i.d. case, but there are more general cases where it also holds (e.g., for $q$-dependent and mixing processes).

To show the CLT, we require a restriction on $\delta$. This is because the bias caused by the noise, which is negligible in Theorem \ref{stochnon}, becomes more substantial when multiplying with the rate of convergence.

\begin{theorem}
\label{cltnon}
Assume that $\mathbb{E} \left(| \epsilon^{j} |^{8} \right) < \infty$ for all $j = 1,...,d$ and $(k_{n}, \theta)$ satisfy Eq. \eqref{deltanon}. As $n \to \infty$, it holds that
\begin{itemize}
\item[(i)] If $\delta > 0.1$
\begin{equation}
\label{clt1non}
n^{1/4 -\delta/2} \left( MRC \left[ Y \right]_n^\delta - \int_0^1 \Sigma _{s} \text{\upshape{d}}s  \right)
\overset{d_{s}}{ \to} MN \left( 0, \frac{2 \Phi_{22} \theta}{ \psi_2^2 } \int_0^1 \Lambda _s \text{\upshape{d}}s \right),
\end{equation}
where $(\Lambda _s)$ is defined in \eqref{ltu}.
\item[(ii)] If $\delta = 0.1$
\begin{equation} \label{clt2non}
n^{1/5} \left( MRC \left[ Y \right]_n^\delta   - \int_0^1 \Sigma_{s} \text{\upshape{d}}s  \right)
\overset{d_{s}}{ \to} MN \left( \frac{ \psi_1}{ \theta^2 \psi_2} \Psi, \frac{2 \Phi_{22} \theta}{ \psi_2^2 } \int_0^1 \Lambda _s \text{\upshape{d}}s \right).
\end{equation}
\end{itemize}
\end{theorem}
\begin{proof}
See appendix. \qed
\end{proof}
Theorem \ref{cltnon} amounts to a classical bias-variance trade-off. It shows the expected result that using a longer pre-averaging window $k_{n} = O(n^{1/2 + \delta})$ averages enough to make $MRC \left[ Y \right]_{n}^{ \delta}$ consistent without a bias-correction, but it also slows down its rate of convergence. The larger is $\delta$, the harder is this effect. Note that the asymptotic variance depends solely on the volatility process $\sigma$, while the two noise-dependent terms (a cross-term and a pure noise term) appearing in Eq. \eqref{L2} are wiped out.

The optimal choice $\delta = 0.1$ results in a $n^{- 1 / 5}$ rate of convergence, and a large sample bias of the order $n^{-1 / 5} \frac{ \psi_{1}}{ \theta^{2} \psi_{2}} \Psi$. If desired, the bias can be estimated using $\hat{ \Psi}_{n}$, but it should be small for more recent data and could be ignored.\footnote{If we do bias-correct $MRC \left[ Y \right]_{n}^{ \delta}$, the resulting estimator is then again not ensured to be positive semi-definite.}

This result bears some resemblance to that of the multivariate realised kernel of \citet*{barndorff-nielsen-hansen-lunde-shephard:08c}. A flat-top kernel estimator can be designed to converge at rate $n^{-1/4}$, but the resulting estimator may go negative. To guarantee positive semi-definiteness, the authors propose kernel corrections that are not entirely flat-top and they show this results in $n^{-1 /5}$ rate of convergence and a non-zero asymptotic mean as produced here as well.

\subsection{Mapping MRC into a realised kernel estimator}
We can indeed make a stronger link between the multi-scale RV, the kernel approach and our pre-averaging method, when estimating quadratic variation. Here, we provide some more details about this relationship. To fix ideas, we take $d = 1$ (i.e., the univariate setting) and compare only to the kernel approach. \citet*{barndorff-nielsen-hansen-lunde-shephard:08a} then show how the multi-scale RV fits into the realised kernel setting.

When we explicitly include the bias-correction and ignore finite sample issues, the $MRC \left[ Y \right]_{n}$ estimator is given by the formula
\begin{equation}
MRC \left[ Y \right]_{n} = \frac{1}{ \psi_{2} k_{n}}  \sum_{i = 0}^{n - k_{n} + 1} | \bar Y_{i}^{n} |^{2} - \frac{ \psi_{1}}{ \theta^{2} \psi_{2}} \hat{ \Psi}_{n}.
\end{equation}
Consider a flat-top kernel-based estimator with kernel weights
\begin{equation*}
k(s) = \frac{ \phi_{2}(s)}{ \psi_{2}},
\end{equation*}
where $\phi_{2}(s)$ and $\psi_{2}$ are defined as in Section \ref{preavg}. We call the function $k$ a flat-top kernel, if (i) $k(0) = 1$, (ii) $k(1) = 0$ and (iii) $k'(0) = k'(1) = 0$. We note that $k(s) = \phi_2(s) / \psi_2$ satisfies all three conditions: (i) and (ii) are trivial, while the third condition follows from the identity:
\begin{equation*}
k'(s) = - \frac{1}{ \psi_{2}} \int_s^{1} g(x) g'(x - s) \text{d}x,
\end{equation*}
and integration by parts (recall that $g(0) = g(1) = 0$).

In the next step, we map the MRC statistic into a kernel-like one as follows
\begin{equation}
MRC \left[ Y \right]_{n} = \sum_{i = 1}^{n} \delta_{0i} | \Delta_i^n Y |^2 + 2 \sum_{h = 1}^{k_n - 2} \sum_{i = 1}^{n - h} \delta_{hi} \Delta_i^n Y \Delta_{i + h}^n Y
\end{equation}
with
\begin{equation}
\delta_{0i} = \left \{
\begin{array} {lc}
\frac{1}{ \psi_{2} k_{n}} \sum_{j = 1}^i g^2 ( \frac{j}{k_{n}}) - \frac{ \psi_{1}}{ \theta^{2} \psi_{2} 2n} :& 1\leq i\leq k_n -2 \\[0.25cm]
\frac{1}{ \psi_{2} k_{n}} \sum_{j = 1}^{k_n-1} g^2 ( \frac{j}{k_{n}}) - \frac{ \psi_{1}}{ \theta^{2} \psi_{2} 2n} :& k_n-1\leq i
\leq n -k_n+2 \\[0.25cm]
\frac{1}{ \psi_{2} k_{n}} \sum_{j = 1}^{n - i +1} g^2 ( \frac{k_{n}-j}{k_{n}}) - \frac{ \psi_{1}}{ \theta^{2} \psi_{2} 2n} :&
n -k_n+3 \leq i
\leq n
\end{array}
\right.
\end{equation}
and
\begin{equation}
\delta_{hi} = \left \{
\begin{array} {lc}
\frac{1}{ \psi_{2} k_{n}} \sum_{j = 1}^i g ( \frac{j}{k_{n}}) g ( \frac{j + h}{k_{n}}) :& 1\leq i\leq k_n -h-2 \\[0.25cm]
\frac{1}{ \psi_{2} k_{n}} \sum_{j = 1}^{k_n-h-1} g ( \frac{j}{k_{n}}) g ( \frac{j + h}{k_{n}}) :& k_n-h-1\leq i \leq n -k_n+2 \\[0.25cm]
\frac{1}{ \psi_{2} k_{n}} \sum_{j = 1}^{n - i +1 } g ( \frac{k_{n} - j}{k_{n}}) g ( \frac{k_{n} - j + h}{k_{n}}) :&
n -k_n+3 \leq i
\leq n
\end{array}
\right.
\end{equation}
for $1\leq h\leq k_n-2$.

An example and some remarks are now in order.

\noindent \textbf{Example} \ Take $g(x) = \min(x, 1 - x)$ for $x\in [0,1]$, which is our canonical weight function. Then, the corresponding kernel $k(s)= \phi_2(s) / \psi_2$ is the Parzen kernel.

\noindent \textbf{Remark} \ Apart from border terms, i.e. terms close to 0 and 1, the pre-averaging estimator coincides with the kernel-based estimator using the flat-top kernel
function $k(s) = \phi_2(s)/\psi_2$. Both estimators possess the optimal rate of convergence $n^{-1/4}$, but the border terms affect their asymptotic distribution. In fact, to arrive at their central limit theorem
for the realised kernel, \citet*{barndorff-nielsen-hansen-lunde-shephard:08c} need to apply some averaging (or "jittering") to edge terms, while the pre-averaging estimator is asymptotically mixed normal "by construction".

\noindent \textbf{Remark} \ Using the definition $k(s)= \phi_2(s)/\psi_2$, we learn that for every weight function $g$ there exists a flat-top kernel $k$. However, the converse statement is not true in general. Thus, the class of flat-top kernels is larger than our class of pre-averaging functions, but this does not appear to be a noticeable disadvantage for practical applications \citep*[][for example, advocate using the Parzen kernel in practice, which corresponds to our $g$ function by the example]{barndorff-nielsen-hansen-lunde-shephard:08c}.

It is worthwhile to underscore that pre-averaging is a general approach, which can be used to estimate various characteristics of semimartingales, when they are cloaked with noise. \citet*{jacod-li-mykland-podolskij-vetter:09a}, for example, pre-average realised variance, but it has also been used in jump-robust inference in conjunction with the bipower variation statistic \citep*[e.g.,][]{podolskij-vetter:09a} or the quantile-based realised variance \citep*[e.g.,][]{christensen-oomen-podolskij:10a}. This is also the reason, why we can estimate the asymptotic conditional variance in the CLT using the same type of estimator to deliver a feasible result. Other estimation methods are typically designed to estimate quadratic variation and cannot, in general, be used to solve other estimation problems.

\subsection{Applying the MRC to non-synchronous data}
Non-synchronous trading has long been a recognized feature of multivariate high-frequency data \citep*[see, e.g.,][]{fisher:66a,lo-mackinlay:90a}. This causes spurious cross-autocorrelation amongst asset price returns sampled at regular intervals in calendar time, as new information gets built into prices at varying intensities. It is well-known that high-frequency realised covariance estimates, using for example the previous-tick method to align prices, are biased towards zero in this setting \citep*[e.g.,][]{epps:79a}.

Motivated by these shortcomings of realised covariance, a number of alternative procedures have been proposed in the literature. \citet*{scholes-williams:77a} and \citet*{dimson:79a} suggested to include leads and lags of sample autocovariances of high-frequency returns into the realised covariance to correct for stale prices. This results in a bias reduction but also increases the variance of the estimator \citep*[e.g.,][]{griffin-oomen:06a}. Importantly, the lead-lag realised covariance is generally inconsistent in the presence of noise. More recently, \citet*{zhang:08a} extends the two-scale RV to the multivariate setting and shows that it is consistent under non-synchronous prices and noise. An analytic derivation of the impact of Epps effect on realised covariance under previous-tick sampling is also given. In independent and concurrent work to ours, \citet*{barndorff-nielsen-hansen-lunde-shephard:08c} propose a multivariate realised kernel. They rely on the concept of refresh time to match prices in calendar time and show that stale pricing errors under refresh time sampling do not change the asymptotic distribution of the multivariate realised kernel.

A related approach can be taken to make the MRC robust to non-synchronous observation times. In particular, under suitable regularity conditions on the sampling times, non-synchronous trading does not alter the asymptotic distribution of the MRC, when constructing synchronous time series of prices. A multivariate time series of high-frequency returns obtained in this fashion can therefore be plugged into the asymptotic theory developed above without concern. The proof of these results are highly technical, but their validity should be clear in the light of the comparison with the multivariate realised kernel given above. Therefore, we omit the detailed proofs of these results. Instead, we conduct some simulations in section \ref{Section:simulations} to verify the correctness of these conjectures.

\subsubsection{A pre-averaged Hayashi-Yoshida estimator}
\citet*{hayashi-yoshida:05a,hayashi-yoshida:08a} develop an alternative procedure for covariance measurement in the noise-free case, which is based on the original non-synchronous data \citep*[see, e.g.,][for related work]{dejong-nijman:97a,martens:03a,palandri:06a,corsi-audrino:07a}. This estimator has the profound advantage that it does not throw away information that is typically lost using a synchronization procedure. Here, we show how this estimator can also be made robust to noise by using pre-averaging.

Given the vector of log-prices $Y = (Y^{1}, \ldots, Y^{d})'$, which is defined by the noisy diffusion model in Eq. \eqref{zsh}, we now assume that the component processes $(Y^{k})$ are observed at non-random time points $t_{i}^{(k)}$, for $i = 1, \ldots, n_{k}$, with $(t_{i}^{(k)})$ being a partition of the interval $[0,1]$ and $k = 1, \ldots, d$. In addition, we need some regularity conditions on the sampling such that all time schemes are comparable in the following way: $\max |t_{i}^{(k)} - t_{i - 1}^{(k)}| \rightarrow 0$ as $n_{k} \to \infty$, for $k = 1, \ldots, d$ and
\begin{equation}
\label{schemereg}
\max_{1 \leq i \leq n_{k}} \# \left\{ t_{j}^{(k)} \mid t_{j}^{(k)} \in \bigl[ t_{i}^{(l)}, t_{i}^{(l)} \bigr] \right\} \leq K,
\end{equation}
for $1 \leq k, l \leq d$ and some $K > 0$, where $K$ is independent of $n_{k}$, for $k = 1, \ldots, d$. The latter condition says that data from one process do not cluster in any single interval of the others. Finally, the following condition is needed
\begin{equation}
\label{schemebound}
\frac{\max|t_i^{(k)} - t_{i-1}^{(k)}|}{\min|t_i^{(k)} - t_{i-1}^{(k)}|}\leq  c~,
\end{equation}
where $c>0$ is a constant independent of $n_k$ for all $1\leq k\leq d$. This condition implies
that $n_k \max|t_i^{(k)} - t_{i-1}^{(k)}|\leq  c$, which could be assumed instead.

A consistent estimator of the integrated covariance $\int_{0}^{1} \Sigma_{s}^{kl} \text{d}s$ between assets $Y^{k}$ and $Y^{l}$ can then be constructed as follows. We set $n = \sum_{k = 1}^{d} n_{k}$ and define the statistic
\begin{equation}
HY[Y]_{n}^{(k,l)} = \frac{1}{ \left( \psi_{HY} k_{n} \right)^{2}} \sum_{i = 0}^{n_{k} - k_{n} + 1} \sum_{j = 0}^{n_{l} - k_{n} + 1} \bar{Y^{k}}_{i}^{n} \bar{Y^{l}}_{j}^{n} \mathbb{I}_{ \{ (t_{i}^{(k)}, t_{i + k_n}^{(k)}] \cap (t_j^{(l)}, t_{j+k_n}^{(l)}] \neq \emptyset \}},
\end{equation}
where $\psi_{HY} = \int_{0}^{1} g(x) \text{d}x$ and $\mathbb{I}_{ \left\{ \bullet \right\} }$ is the indicator function discarding pre-averaged returns that do not overlap in time. Hence, $HY[Y]_n^{(k,l)}$ is a pre-averaged version of the \citet*{hayashi-yoshida:05a} estimator. Note that under the previous conditions, $n$, $n_{k}$ and $n_{l}$ are of the same order and that $n$ controls the universal pre-averaging window $k_{n}$.\footnote{Our choice here implies that $k_{n}$ is identical across all pairs of asset combinations. In general, the best way of choosing $k_{n}$ depends on what is being estimated (e.g., covariance, correlation or beta). Thus, if one only cares about efficiency in pairwise covariance estimation, a better $HY[Y]_n^{(k,l)}$ estimator might be based on a local $k_n$, for example $k_{n}^{k,l} = \theta_{k,l} \sqrt{n} + o(n^{1/2})$. This would also serve to make the estimated covariances universe independent (i.e., they will not change by adding or removing assets). To make a more qualified, statistical argument on this issue, however, we need to compare an expression for the asymptotic variance of $HY[Y]_{n}^{(k,l)}$ under different ways of choosing $k_n$ (e.g., from a CLT). This is beyond the scope of this paper and will be left for future research, where we hope to shed more light on the subject.}

\begin{theorem}
\label{plimHYn} Assume that $\mathbb{E} \left( | \epsilon^j | ^{4} \right)  < \infty$ for all $j = 1, \ldots, d$, $(k_{n}, \theta)$ satisfy Eq. \eqref{kthetarelation} and $g(x) > 0$ for
$x \in (0,1)$. As $n \to \infty$, it holds that
\begin{equation}
HY[Y]_{n}^{(k,l)} \overset{p}{ \to} \int_{0}^{1} \Sigma_{s}^{kl} \text{\upshape{d}}s,
\end{equation}
for $1 \leq k, l \leq d$.
\end{theorem}
\begin{proof}
See appendix. \qed
\end{proof}

Interestingly, Theorem \ref{plimHYn} shows the somewhat surprising result that there is no asymptotic noise-induced bias in $HY[Y]_{n}^{(k,l)}$, not even when the spacings $(t_{i}^{(k)})$ and $(t_{j}^{(l)})$ are identical. This can be seen as follows. First, under the i.i.d. structure on the noise only products of the form $\epsilon_{t_{i}^{(k)}}^{k} \epsilon_{t_{j}^{(l)}}^{l}$ with $t_{i}^{(k)} = t_{j}^{(l)} = t$ contribute to potential bias. We consider that set of points and assume that $k_{n} \leq i \leq n_{1} - k_{n}$ and $k_{n} \leq j \leq n_{2} - k_{n}$ (this is innocent, for the summands which do not fulfill this are negligible). Then, an inspection of $HY[Y]_{n}^{(k,l)}$ shows that all products
$\epsilon_{t}^{k} \epsilon_{t}^{l}$ appear with the factor
\begin{equation*}
\left( \sum_{j = 0}^{k_{n} - 1} g( \frac{j + 1}{k_{n}} ) - g( \frac{j}{k_{n}}) \right)^{2}
\end{equation*}
in front. But $\sum_{j = 0}^{k_{n} - 1} g( \frac{j + 1}{k_{n}} ) - g( \frac{j}{k_{n}}) = 0$, because $g(0) = g(1) = 0$, so these terms drop out of the summation.

While $HY[Y]_{n}^{(k,l)}$ has the advantage that it is free of prior alignment of log-prices and hence does not throw away information in the sample, it does suffer from being a pairwise estimator, which means that once we assemble all the single variance/covariance estimates into a full covariance matrix, it is not guaranteed to be positive semi-definite. Still, there are some problems in financial economics, in which it is only the accuracy of the estimator that matters and positive semi-definiteness is less important, for example in asset allocation and risk management under gross exposure constraints \citep*[e.g.][]{fan-zhang-yu:09a}. Moreover, our empirical results show that $(HY[Y]_{n}^{(k,l)})_{1 \leq k,l \leq d}$ does not fail to be positive semi-definite on a single day for the $d = 5$-dimensional vector of asset prices considered there.

\begin{proposition}
\label{orderTn}
Assume that $\mathbb{E} \left( | \epsilon^{j} |^{4} \right) < \infty$ for all $j = 1, \ldots, d$. If $(k_{n}, \theta)$ satisfy Eq. \eqref{kthetarelation} and $g(x)>0$ for $x \in (0,1)$, then
\begin{equation*}
\text{\upshape{var}} ( HY[Y]_{n}^{(k,l)} ) = O ( n^{- 1 / 2} ),
\end{equation*}
that is, $HY[Y]_{n}^{(k,l)}$ has the optimal rate of convergence.
\end{proposition}
\begin{proof}
See appendix. \qed
\end{proof}

Proposition \ref{orderTn} shows that the rate of convergence associated with $HY[Y]_{n}^{(k,l)}$ is $n^{-1/4}$. A proof of the much stronger result, the CLT for $HY[Y]_{n}^{(k,l)}$, will be given in a companion paper to this one \citep*[see][]{christensen-podolskij-vetter:10a}. In the empirical section, we gauge the properties of this estimator on actual data and find that the pre-averaged version performs very well.

\section{An estimator of the asymptotic covariance matrix}
\label{jetzt} In order to make Theorem \ref{clt2} and \ref{cltnon} feasible, we need to estimate the asymptotic covariance matrix $\text{avar}_{\text{MRC}}$, as it appears in Eq. \eqref{L2}. Here, we give an explicit estimator of $\text{avar}_{\text{MRC}}$. More precisely, we present an estimator of the asymptotic covariance matrix of the vectorized statistic in Eq. \eqref{mainstats}.

First, we set
\begin{equation}
\chi_{i}^{n} = \text{vec} \Bigl( \bar{Y}_{i}^{n} ( \bar{Y}_{i}^{n})' \Bigr),
\end{equation}
where $\text{vec}(\cdot)$ is the vectorisation operator that stacks columns of a matrix below one another. Next, we define the statistic
\begin{equation*}
V_n(g)= \sum_{i=0}^{n-k_n+1} \chi_i^n (\chi_i^n)' - \frac{1}{2} \sum_{i=0}^{n-2k_n+1}
\Big( \chi_i^n (\chi_{i+k_n}^n)' + \chi_{i+k_n}^n (\chi_i^n)' \Big).
\end{equation*}
We should note that $V_n(g)$ depends on both the bandwidth parameter $\theta$ and the pre-averaging function $g$, and that it is positive semi-definite by construction. Moreover, for any $1\leq k,k',l,l'\leq d$, we get the following convergence
\begin{equation*}
V_n^{(k-1)d + k', (l-1)d + l'}(g) \overset{p}{ \to} a_B(g,\theta ) \int_0^1 \Lambda_u^{kk',ll'} \text{d}u + a_M(g,\theta ) \int_0^1 \Theta_u^{kk',ll'} \text{d}u
+a_N(g,\theta ) \Upsilon^{kk',ll'},
\end{equation*}
where $a_B(g,\theta )=\theta^2 \psi_{2}^2$, $a_M(g,\theta )= \psi_{1} \psi_{2}$ and $a_N(g,\theta ) =
\frac{ \psi_{1}^2}{ \theta^2}$ \citep*[the proof of this result is achieved by using arguments alike the ones presented in the proof of Theorem \ref{consistency}; see][for further details]{kinnebrock-podolskij:08b}.

What needs to be estimated is
\begin{equation*}
\text{avar}_{\text{MRC}} = \frac{2}{\psi _2^2 }\left( \Phi _{22} \theta  \int_0^1 \Lambda _u \text{d}u  +
\frac{\Phi_{12}}{\theta }\int_0^1 \Theta _u \text{d}u  + \frac{\Phi _{11}}{\theta ^3}\Upsilon  \right),
\end{equation*}
with all constants referring to a given function $g_0$. Suppose that $g_{0}(x) = \min(x,1 - x)$ as above. Now, we take three different functions $g_1$, $g_2$ and $g_3$ that satisfy all the conditions of $g_0$, and which are chosen such that the matrix
\begin{equation*}
A(g_1,g_2,g_3)= \left( \begin{array} {ccc}
a_B(g_1,\theta ) & a_M(g_1,\theta ) & a_N(g_1,\theta ) \\
a_B(g_2,\theta ) & a_M(g_2,\theta ) & a_N(g_2,\theta ) \\
a_B(g_3,\theta ) & a_M(g_3,\theta ) & a_N(g_3,\theta )
\end{array}
\right)
\end{equation*}
is invertible. We can then construct a weight vector
\begin{equation}
\label{Eqn:Cofg}
C(g_1,g_2,g_3)= \Big(\frac{2 \Phi _{22} \theta}{\psi _2^2 }, \frac{2 \Phi _{12} }{\psi _2^2 \theta}, \frac{2 \Phi _{11} }{\psi _2^2 \theta^3}  \Big) A^{-1} (g_1,g_2,g_3).
\end{equation}
Finally, we form the statistic
\begin{equation}
\label{Hn2}
\widehat{\text{avar}}_{\text{MRC},n} = C^{(1)}(g_1,g_2,g_3) V_n(g_1) + C^{(2)}(g_1,g_2,g_3) V_n(g_2) + C^{(3)}(g_1,g_2,g_3) V_n(g_3),
\end{equation}
where $V_n(g_k)$ are the above estimators associated with the functions $g_k, k = 1,2,3$. Then, it holds that
\begin{equation*}
\widehat{\text{avar}}_{\text{MRC},n}^{(k-1)d + k', (l-1)d + l'} \overset{p}{ \to} \text{avar}_{\text{MRC}}^{kk',ll'},
\end{equation*}
for any $1\leq k,k',l,l'\leq d$.\footnote{The weights $\Big(\frac{2 \Phi _{22} \theta}{\psi _2^2 }, \frac{2 \Phi _{12} }{\psi _2^2 \theta}, \frac{2 \Phi _{11} }{\psi _2^2 \theta^3}  \Big)$ in Eq. \eqref{Eqn:Cofg} are the coefficients in front of $(\int_0^1 \Lambda_u \text{d}u, \int_0^1 \Theta _u \text{d}u, \Upsilon)$ in the expression for the asymptotic variance of MRC. This choice results in a consistent estimator of $\text{avar}_{\text{MRC}}$. More generally, one can estimate any linear combination of $(\int_0^1 \Lambda_u \text{d}u, \int_0^1 \Theta _u \text{d}u, \Upsilon)$ by choosing an appropriate weight vector in Eq. \eqref{Eqn:Cofg}. For example, the multivariate version of the integrated quarticity, $\int_0^1 \Lambda_u \text{d}u$, can be estimated by plugging the linear combination $(1,0,0) A^{-1} (g_1,g_2,g_3)$ into Eq. \eqref{Hn2}.}

There are various classes of functions from which the $g_k$ can be selected, for example $g(x) = x^{a} \left( 1 - x \right)^{b}$ with $a,b \geq  1$ or $g(x) = \sin(c \pi x)$ for integer $c$.
If one manages to find a set of functions $g_k$ for which the coefficients $C^{(k)}(g_1,g_2,g_3)$ are all positive, $k = 1, 2$ and $3$, then the statistic in Eq. \eqref{Hn2} is guaranteed to be positive semi-definite, because it is a linear combination of positive semi-definite statistics $V_n(g_k)$ using positive weights $C^{(k)}(g_1,g_2,g_3)$. It appears to be quite hard to find such a combination in practice, and so far we did not succeed at this. Nonetheless, $\widehat{\text{avar}}_{\text{MRC},n}$ remains a consistent estimator.

\section{Asymptotic theory for covariance, regression and correlation}
The results in section 3 and 4 can be applied in order to compute confidence intervals for some functionals of $\int_{0}^{1} \Sigma_{u} \text{d}{u}$ that are important in practice, such as covariance, regression and correlation. For the $i$th and $j$th asset, these quantities are given by
\begin{equation} \label{covcorr}
\int_0^1 \Sigma_u^{ij} \text{d}u, \quad \beta ^{(ji)}  = \frac{ \int_0^1 \Sigma^{ij}_u \text{d}u}{ \int_0^1 \Sigma ^{ii}_u \text{d}u},
\quad \rho^{(ji)} = \frac{ \int_0^1 \Sigma ^{ij}_u \text{d}u}{ \sqrt {\int_0^1 \Sigma ^{ii}_u \text{d}u  \int_0^1 \Sigma ^{jj}_u \text{d}u}}.
\end{equation}
Theorem \ref{consistency}, \ref{stochnon} or \ref{plimHYn} can be invoked to provide consistent
estimates of $\int_0^1 \Sigma_u^{ij} \text{d}u$, $\beta ^{(ji)}$ and $\rho^{(ji)}$, e.g. for $MRC \left[Y \right]_{n}$
\begin{equation*}
MRC \left[ Y \right]_{n}^{i, j} \overset{p}{ \to} \int_0^1 \Sigma_u^{ij} \text{d}u, \quad \hat \beta^{(ji)}_n = \frac{MRC \left[ Y \right]_{n}^{i, j}}{
MRC \left[ Y \right]_{n}^{i, i}} \overset{p}{ \to} \beta ^{(ji)},
\end{equation*}
\begin{equation}
\label{estcorr}
\hat \rho^{(ji)}_n = \frac{MRC \left[ Y \right]_{n}^{i, j}}{ \sqrt{MRC \left[ Y \right]_{n}^{i, i} MRC  \left[ Y \right]_{n}^{j, j}}} \overset{p}{ \to} \rho^{(ji)}
\end{equation}
for any $1\leq i,j\leq d$. The estimators in \eqref{estcorr} are called modulated realised covariance, regression and correlation, respectively. In the next theorem we present the associated feasible central limit theorems, which follow
from Theorem \ref{clt2} and the delta method for stable convergence.
\begin{theorem} \label{cltcov}
Assume that $\mathbb{E} \left( | \epsilon^{j} |^{8} \right) < \infty$ for all $j = 1,...,d$ and $(k_n, \theta)$ satisfy Eq. \eqref{kthetarelation}. As $n \to \infty$, it holds that
\begin{eqnarray}
&& \frac{n^{1 / 4} \Big( MRC \left[ Y \right]_{n}^{i, j} - \int_0^1 \Sigma_u^{ij} \text{\upshape{d}}u \Big)}{\widehat{\text{avar}}_{\text{MRC},n}^{(i-1)d + j, (i-1)d + j}}
\overset{d}{ \to} N(0,1), \\[0.25cm]
&& \frac{n^{1/4} (\hat \beta_n^{(ji)} - \beta^{(ji)} )}{ \sqrt{ \left( MRC  \left[ Y \right]_{n}^{i,i} \right)^{- 2}  g^{(ji)}_n } }
\overset{d}{ \to} N(0,1), \\[0.25cm]
&& \frac{n^{1/4} (\hat \rho ^{(ji)}_n - \rho^{(ji)})}{ \sqrt{\left( MRC \left[ Y \right]_{n}^{i, i} MRC \left[ Y \right]_{n}^{j, j} \right)^{-1} h^{(ji)}_{n}}} \overset{d}{ \to} N(0,1),
\end{eqnarray}
where $\widehat{\text{avar}}_{\text{MRC},n}$ is given in \eqref{Hn2} and $g^{(ji)}_n$ and $h^{(ji)}_n$ are defined by
\begin{equation*}
g^{(ji)}_n= \Big( 1, - \hat \beta_n^{(ji)} \Big) \Gamma_n \Big( 1, - \hat \beta_n^{(ji)} \Big)'~, \qquad
h^{(ji)}_n= \Big( - \frac{1}{2} \hat \beta_n^{(ji)}, 1,  - \frac{1}{2} \hat \beta_n^{(ij)}\Big)
\overline \Gamma_n \Big( - \frac{1}{2} \hat \beta_n^{(ji)}, 1,  - \frac{1}{2} \hat \beta_n^{(ij)}\Big)'
\end{equation*}
with
\begin{eqnarray*}
\Gamma_n = \left( \begin{array} {cc}
\widehat{\text{avar}}_{\text{MRC},n}^{(i-1)d + j, (i-1)d + j} & \widehat{\text{avar}}_{\text{MRC},n}^{(i-1)d + j, (i-1)d + i} \\
\bullet & \widehat{\text{avar}}_{\text{MRC},n}^{(i-1)d + i, (i-1)d + i}
\end{array}
\right)~, \\ [2.0 ex]
\overline \Gamma_n = \left( \begin{array} {ccc}
\widehat{\text{avar}}_{\text{MRC},n}^{(i-1)d + i, (i-1)d + i} & \widehat{\text{avar}}_{\text{MRC},n}^{(i-1)d + j, (i-1)d + i} & \widehat{\text{avar}}_{\text{MRC},n}^{(i-1)d + i, (j-1)d + j} \\
\bullet & \widehat{\text{avar}}_{\text{MRC},n}^{(i-1)d + j, (i-1)d + j} & \widehat{\text{avar}}_{\text{MRC},n}^{(i-1)d + j, (j-1)d + j} \\
\bullet & \bullet & \widehat{\text{avar}}_{\text{MRC},n}^{(j-1)d + j, (j-1)d + j}
\end{array}
\right).
\end{eqnarray*}
\end{theorem}
All the required terms are easy to compute, so it is rather simple to implement the estimators.

\section{Simulation study}
\label{Section:simulations}
We now demonstrate the finite sample accuracy of some of the asymptotic results developed above using simulations. The design of our Monte Carlo study, which we briefly describe here, is identical to the analysis conducted in \citet*{barndorff-nielsen-hansen-lunde-shephard:08c}.

Specifically, to simulate log-prices we consider the following bivariate stochastic volatility model
\begin{equation}
\text{d}X_{t}^{(i)} = a^{(i)} \text{d}t + \rho^{(i)} \sigma_{t}^{(i)} \text{d}B_{t}^{(i)} + \sqrt{1 - [\rho^{(i)}]^2} \sigma_{t}^{(i)} \text{d}W_{t}, \quad \text{ for } i = 1,2,
\end{equation}
where $B^{(i)}$ and $W$ are independent Brownian motions. In this model, the term $\rho^{(i)} \sigma_{t}^{(i)} \text{d}B_{t}^{(i)}$ is an idiosyncratic component, while $\sqrt{1 - [\rho^{(i)}]^2} \sigma_{t}^{(i)} \text{d}W_{t}$ is a common factor.

The spot volatility is modeled as $\sigma_{t}^{(i)} = \exp (\beta_0^{(i)} + \beta_1^{(i)} \varrho_{t}^{(i)})$ with an Ornstein-Uhlenbeck specification for $\varrho_{t}^{(i)}$: $\text{d} \varrho_{t}^{(i)} = \alpha^{(i)} \varrho_{t}^{(i)} \text{d}t + \text{d}B_{t}^{(i)}$. This implies that there is perfect correlation between the innovations of $\rho^{(i)} \sigma_{t}^{(i)} \text{d}B_{t}^{(i)}$ and $\sigma_{t}^{(i)}$, while it is $\rho^{(i)}$ between the increments of $X_{t}^{(i)}$ and $\varrho_{t}^{(i)}$. Finally, the magnitude of correlation between the two underlying price processes $X_{t}^{(1)}$ and $X_{t}^{(2)}$ is $\sqrt{1 - [\rho^{(1)}]^2} \sqrt{1 - [\rho^{(2)}]^2}$. The reported results are based on the following configuration of parameters for both processes: $(a^{(i)}, \beta_0^{(i)}, \beta_1^{(i)}, \alpha^{(i)}, \rho^{(i)}) = (0.03, -5/16, 1(8,-1/40,-0.3)$, so that $\beta_0^{(i)} = [\beta_1^{(i)}]^{2} / [2 \alpha^{(i)}]$. We note that this particular choice of parameters also means that the volatility process has been normalized, in the sense that $\mathbb{E} ( \int_0^1 [\sigma_s^{(i)}]^2 \text{d}s) = 1$.

We simulate 1,000 paths of this model over the interval $[0,1]$. Motivated by our empirical data, we let $[0,1]$ represent 6.5 hours worth of trading, which is then further decomposed into $N = 23,400$ subintervals of equal length $1 / N$ ($N$ denoting the number of seconds in 6.5 hours). In constructing noisy prices $Y^{(i)}$, we first generate a complete high-frequency record of $N$ equidistant observations of the efficient price $X^{(i)}$ using a standard Euler scheme. The initial values for the $\varrho_{t}^{(i)}$ processes at each simulation run is drawn randomly from the stationary distribution of $\varrho_{t}^{(i)}$, which is $\varrho_{t}^{(i)} \sim N(0,[-2\alpha^{(i)}]^{-1})$.\footnote{Note that the Ornstein-Uhlenbeck process admits an exact discretization \citep*[see, e.g.,][]{glassermann:04a}. We use that result here to avoid discretization errors in approximating the continuous time distribution of $\text{d} \varrho^{(i)}$ over discrete time steps of size $\Delta t = 1 / N$.}

We add simulated microstructure noise $Y^{(i)} = X^{(i)} + \epsilon^{(i)}$ by taking
\begin{equation}
\epsilon^{(i)} \mid \{ \sigma, X \} \ \overset{\text{i.i.d}}{\sim} \ N(0, \omega^2) \quad \text{with} \quad \omega^{2} = \gamma^2 \sqrt{ \frac{1}{N} \sum_{j=1}^N \sigma_{j / N}^{(i)4}}.
\end{equation}
This choice again follows \citet*{barndorff-nielsen-hansen-lunde-shephard:08c} and means that the variance of the noise process increases with the level of volatility of $X^{(i)}$, as documented by \citet*{bandi-russell:06a}. $\gamma^{2}$ takes the values 0, 0.001, 0.01, which covers scenarios with no noise through low-to-high levels of noise.

Finally, we extract irregular, non-synchronous data from the complete high-frequency record using Poisson process sampling to generate actual observation times, $\{ t_{j}^{(i)} \}$. In particular, we consider two independent Poisson processes with intensity parameter $\lambda = \left( \lambda_1, \lambda_2 \right)$. Here $\lambda_{i}$ denotes the average waiting time (in seconds) for new data from process $Y^{(i)}$, so that an average day will have $N / \lambda_{i}$ observations of $Y^{(i)}, i = 1,2$. We vary $\lambda_{1}$ through $(3, 5, 10, 30, 60)$ to capture the influence of liquidity on the performance of our estimators, and we set $\lambda_{2} = 2 \lambda_{1}$ such that on average $Y^{(2)}$ refreshes at half the pace of $Y^{(1)}$.

\begin{table}[H]
\setlength{\tabcolsep}{0.25cm}
\renewcommand{\arraystretch}{0.25}
\begin{center}
\caption{Simulation results}
\label{table:simulation}
\medskip
\begin{small}
\begin{tabular}{lrrrrrrrrrrrrr}
\hline
\multicolumn{5}{l}{\emph{Panel A: Integrated covariance}}\\
 & \multicolumn{2}{c}{$Cov^{15m}$}
 & \multicolumn{2}{c}{$Cov^{1m}$}
 & \multicolumn{2}{c}{$MRC[Y]_{n=RT}$} & \multicolumn{2}{c}{$MRC[Y]_{n=RT}^{\delta=0.1}$} & \multicolumn{2}{c}{$HY[Y]_{n}^{(k,l)}$}\\
$\xi^2 = 0$ & bias & rmse & bias & rmse & bias & rmse & bias & rmse & bias & rmse\\
$\lambda = (3,6)$ & -0.015 & 0.159 & -0.052 & 0.103 & -0.007 & 0.156 & -0.018 & 0.222 & -0.044 & 0.278\\
$\lambda = (5,10)$ & -0.025 & 0.170 & -0.090 & 0.147 & -0.010 & 0.167 & -0.023 & 0.233 & -0.055 & 0.296\\
$\lambda = (10,20)$ & -0.044 & 0.174 & -0.178 & 0.266 & -0.017 & 0.207 & -0.027 & 0.270 & -0.069 & 0.338\\
$\lambda = (30,60)$ & -0.128 & 0.255 & -0.387 & 0.547 & -0.028 & 0.259 & -0.038 & 0.326 & -0.101 & 0.426\\
$\lambda = (60,120)$ & -0.228 & 0.367 & -0.505 & 0.702 & -0.038 & 0.300 & -0.047 & 0.370 & -0.117 & 0.512\\
$\xi^2 = 0.001$\\
$\lambda = (3,6)$ & -0.019 & 0.166 & -0.054 & 0.122 & -0.008 & 0.148 & -0.018 & 0.220 & -0.044 & 0.275\\
$\lambda = (5,10)$ & -0.025 & 0.174 & -0.089 & 0.161 & -0.011 & 0.176 & -0.022 & 0.239 & -0.053 & 0.298\\
$\lambda = (10,20)$ & -0.046 & 0.188 & -0.183 & 0.292 & -0.016 & 0.198 & -0.028 & 0.267 & -0.071 & 0.345\\
$\lambda = (30,60)$ & -0.125 & 0.246 & -0.383 & 0.538 & -0.030 & 0.263 & -0.036 & 0.329 & -0.099 & 0.431\\
$\lambda = (60,120)$ & -0.226 & 0.368 & -0.506 & 0.705 & -0.039 & 0.295 & -0.043 & 0.366 & -0.122 & 0.518\\
$\xi^2 = 0.01$\\
$\lambda = (3,6)$ & -0.014 & 0.350 & -0.048 & 0.519 & -0.007 & 0.147 & -0.019 & 0.216 & -0.044 & 0.276\\
$\lambda = (5,10)$ & -0.019 & 0.350 & -0.070 & 0.485 & -0.010 & 0.174 & -0.021 & 0.237 & -0.054 & 0.295\\
$\lambda = (10,20)$ & -0.056 & 0.346 & -0.192 & 0.604 & -0.017 & 0.202 & -0.029 & 0.264 & -0.069 & 0.341\\
$\lambda = (30,60)$ & -0.116 & 0.365 & -0.392 & 0.682 & -0.029 & 0.243 & -0.040 & 0.323 & -0.099 & 0.430\\
$\lambda = (60,120)$ & -0.225 & 0.494 & -0.508 & 0.767 & -0.039 & 0.324 & -0.042 & 0.386 & -0.117 & 0.523\\
\emph{Panel B: Integrated correlation}\\
$\xi^2 = 0$ & bias & rmse & bias & rmse & bias & rmse & bias & rmse & bias & rmse\\
$\lambda = (3,6)$ & -0.016 & 0.028 & -0.069 & 0.070 & -0.001 & 0.017 & -0.002 & 0.026 & -0.012 & 0.051\\
$\lambda = (5,10)$ & -0.026 & 0.037 & -0.119 & 0.121 & -0.002 & 0.020 & -0.003 & 0.030 & -0.014 & 0.059\\
$\lambda = (10,20)$ & -0.051 & 0.060 & -0.237 & 0.239 & -0.002 & 0.024 & -0.004 & 0.036 & -0.020 & 0.076\\
$\lambda = (30,60)$ & -0.157 & 0.167 & -0.514 & 0.517 & -0.003 & 0.034 & -0.008 & 0.047 & -0.032 & 0.115\\
$\lambda = (60,120)$ & -0.287 & 0.299 & -0.672 & 0.674 & -0.006 & 0.044 & -0.013 & 0.061 & 0.041 & 1.328\\
$\xi^2 = 0.001$\\
$\lambda = (3,6)$ & -0.141 & 0.149 & -0.438 & 0.440 & -0.001 & 0.018 & -0.005 & 0.027 & -0.012 & 0.050\\
$\lambda = (5,10)$ & -0.146 & 0.154 & -0.463 & 0.465 & -0.002 & 0.020 & -0.005 & 0.030 & -0.014 & 0.060\\
$\lambda = (10,20)$ & -0.170 & 0.179 & -0.528 & 0.530 & -0.003 & 0.025 & -0.007 & 0.037 & -0.019 & 0.077\\
$\lambda = (30,60)$ & -0.257 & 0.266 & -0.657 & 0.659 & -0.005 & 0.035 & -0.012 & 0.050 & -0.032 & 0.118\\
$\lambda = (60,120)$ & -0.368 & 0.378 & -0.738 & 0.739 & -0.007 & 0.046 & -0.017 & 0.062 & -0.033 & 0.458\\
$\xi^2 = 0.01$\\
$\lambda = (3,6)$ & -0.559 & 0.571 & -0.813 & 0.815 & -0.002 & 0.026 & -0.024 & 0.039 & -0.011 & 0.051\\
$\lambda = (5,10)$ & -0.564 & 0.574 & -0.817 & 0.819 & -0.003 & 0.030 & -0.027 & 0.044 & -0.014 & 0.063\\
$\lambda = (10,20)$ & -0.579 & 0.590 & -0.831 & 0.833 & -0.005 & 0.037 & -0.034 & 0.055 & -0.020 & 0.079\\
$\lambda = (30,60)$ & -0.620 & 0.632 & -0.851 & 0.853 & -0.007 & 0.053 & -0.044 & 0.074 & -0.035 & 0.123\\
$\lambda = (60,120)$ & -0.657 & 0.667 & -0.859 & 0.861 & -0.011 & 0.067 & -0.054 & 0.091 & 0.010 & 1.296\\
\emph{Panel C: Integrated beta}\\
$\xi^2 = 0$ & bias & rmse & bias & rmse & bias & rmse & bias & rmse & bias & rmse\\
$\lambda = (3,6)$ & -0.020 & 0.092 & -0.090 & 0.128 & -0.000 & 0.065 & -0.002 & 0.098 & -0.026 & 0.179\\
$\lambda = (5,10)$ & -0.034 & 0.100 & -0.159 & 0.219 & -0.001 & 0.076 & -0.002 & 0.112 & -0.030 & 0.214\\
$\lambda = (10,20)$ & -0.069 & 0.143 & -0.316 & 0.427 & -0.003 & 0.092 & -0.003 & 0.135 & -0.043 & 0.266\\
$\lambda = (30,60)$ & -0.212 & 0.316 & -0.693 & 0.926 & -0.001 & 0.129 & -0.005 & 0.178 & -0.069 & 0.385\\
$\lambda = (60,120)$ & -0.384 & 0.538 & -0.902 & 1.190 & -0.001 & 0.170 & -0.008 & 0.222 & 0.015 & 2.503\\
$\xi^2 = 0.001$\\
$\lambda = (3,6)$ & -0.188 & 0.277 & -0.586 & 0.777 & 0.000 & 0.067 & -0.004 & 0.099 & -0.026 & 0.177\\
$\lambda = (5,10)$ & -0.195 & 0.297 & -0.623 & 0.832 & -0.002 & 0.076 & -0.004 & 0.112 & -0.028 & 0.217\\
$\lambda = (10,20)$ & -0.226 & 0.325 & -0.716 & 0.956 & -0.003 & 0.093 & -0.007 & 0.136 & -0.039 & 0.263\\
$\lambda = (30,60)$ & -0.354 & 0.511 & -0.903 & 1.200 & -0.004 & 0.131 & -0.010 & 0.183 & -0.068 & 0.389\\
$\lambda = (60,120)$ & -0.495 & 0.680 & -1.005 & 1.330 & -0.006 & 0.169 & -0.014 & 0.215 & -0.066 & 0.829\\
$\xi^2 = 0.01$\\
$\lambda = (3,6)$ & -0.747 & 1.010 & -1.093 & 1.450 & 0.000 & 0.091 & -0.031 & 0.116 & -0.025 & 0.185\\
$\lambda = (5,10)$ & -0.758 & 1.033 & -1.096 & 1.442 & -0.002 & 0.114 & -0.035 & 0.134 & -0.032 & 0.220\\
$\lambda = (10,20)$ & -0.772 & 1.047 & -1.123 & 1.501 & -0.009 & 0.137 & -0.045 & 0.159 & -0.045 & 0.289\\
$\lambda = (30,60)$ & -0.839 & 1.149 & -1.154 & 1.530 & -0.006 & 0.181 & -0.059 & 0.215 & -0.076 & 0.380\\
$\lambda = (60,120)$ & -0.888 & 1.203 & -1.169 & 1.559 & -0.005 & 0.222 & -0.065 & 0.239 & -0.055 & 0.787\\
\hline
\end{tabular}\medskip
\end{small}
\begin{footnotesize}
\parbox{.95\textwidth}{\emph{Note}. This table shows the results of the simulation analysis.}
\end{footnotesize}
\end{center}
\end{table}

In Table \ref{table:simulation}, we report the results of the simulations. As the results for estimating the variance components of the $2 \times 2$ covariance matrix are as expected compared to prior work, we give focus here towards estimating the integrated covariance, correlation and beta. Also, because the refresh time sampled MRC estimators perform somewhat better than the previous-tick based MRC estimators, we only report the results for the MRC estimators based on refresh time sampling ($n = RT$).\footnote{All unreported results are available upon request.} In the three panels of the table, we therefore provide the bias and root mean squared error (rmse) of the various estimators in terms of estimating the integrated covariance, correlation and beta. We compare our results to a standard realised covariance sampled at either a 1-minute or 15-minute frequency.

Looking at the table, we see that the MRC estimators are very efficient across all scenarios of noise and non-synchronous trading considered here, matching or outperforming the standard realised covariance. The $HY[Y]^{(k,l)}$ estimator is less efficient, owing largely to a larger finite sample bias in this estimator. We will study the possibilities of making finite sample adjustments to $HY[Y]^{(k,l)}$, elsewhere.

Above, we conjectured without a proof that the MRC was robust to stale prices and retained its rate of convergence under non-synchronous trading. If this is to be true, and ignoring finite sample biases, we should then expect the rmse of the two MRC estimators to decrease at rate $n^{-1/4}$ and $n^{-1/5}$. As can be seen from the table, this is exactly what we find. For example, when estimating the integrated covariance and going from $\lambda = (60,120)$ to $\lambda = (3,6)$, i.e. an approximate twenty-fold increase in sample size, the rmse of $MRC[Y]_n$ decreases roughly by half and the rmse of $MRC[Y]_n^{\delta}$ by slightly less than a half, which is consistent with the rates given above.

All in all, the simulation results show that the estimators proposed in this paper are very good at estimating the integrated covariance, correlation and beta across a wide range of noise and liquidity scenarios, and that the asymptotic predictions given above are also reasonable guides to their finite sample behavior.

\section{Empirical Illustration}
To illustrate some empirical features of the pre-averaging theory developed above, we retrieved high-frequency data for a five-dimensional vector of assets from Wharton Research Data Services (WRDS). We picked four equities at random from the S\&P 500 constituents list as of July 1, 2009. We then added a 5th element, namely the S\&P 500 Depository Receipt (ticker symbol SPY), which is an exchange-traded fund that tracks the large-cap segment of the U.S. stock market. As such, it can be viewed as generating market-wide index returns. The four remaining stocks are the following (with ticker symbol and industry classification in parenthesis): Bristol-Myers Squibb (BMY, health care), Lockheed Martin (LMT, industrials), Oracle (ORCL, information technology) and Sara Lee (SLE, consumer staples), thus representing a broad category of industries. We use both trades and quotes data for the sample period that covers the whole of 2006, which results in 251 trading days.

\begin{table}[H]
\setlength{\tabcolsep}{0.50cm}
\begin{center}
\caption{Descriptive statistics and number of data before and after filtering}
\label{table:descriptive-statistics}
\medskip
\begin{small}
\begin{tabular}{lrrrrr}
\hline
Stock & BMY & LMT & ORCL & SLE & SPY\\ \hline
Exchange & N & N & Q & N & P\\[0.25cm]
\emph{Panel A: Transaction data}\\
Raw trades & 1085420 & 811607 & 8153413 & 533517 & 7685215\\
\hspace*{0.50cm} Corrected/Abnormal/Zeros &    111 &     61 &  13806 &     85 &   2584\\
\hspace*{0.50cm} Time aggregation & 231826 & 139181 & 6599286 &  78039 & 6000898\\ \cline{2-6}
\#Trades & 853483 & 672365 & 1540321 & 455393 & 1681733\\
Intensity &   3400 &   2679 &   6137 &   1814 &   6700\\
Noise ratio, $\gamma$ &  0.363 &  0.336 &  0.484 &  0.656 &  0.202\\
\\
\emph{Panel B: Quotation data}\\
Raw quotes & 5402607 & 3245315 & 23411495 & 3208830 & 17536447\\
\hspace*{0.50cm} Negative/Wide/Zeros &    643 &   3917 &   2623 &    604 &    256\\
\hspace*{0.50cm} Time aggregation & 2547054 & 1075914 & 20050224 & 979299 & 12851464\\ \cline{2-6}
\#Quotes & 2854910 & 2165484 & 3358648 & 2228927 & 4684727\\
Intensity &  11374 &   8627 &  13381 &   8880 &  18664\\
Avg. spread (in cents) &  1.273 &  2.389 &  1.017 &  1.215 &  1.575\\
Noise ratio, $\gamma$ &  0.205 &  0.219 &  0.203 &  0.310 &  0.109\\
\hline
\end{tabular}\medskip
\end{small}
\begin{footnotesize}
\parbox{.95\textwidth}{\emph{Note}. This table reports some descriptive statistics and liquidity measures for the selection of stocks included in our empirical application. We show the exchange from which data are extracted. The exchange code is: N = NYSE, Q = NASDAQ and P = Pacific. Raw trades/quotes is the total number of data available from these exchanges during the trading session, while \# trades/quotes is the total sample remaining after filtering the data. Intensity is the average number of data pr. day, while the noise ratio is defined in \citet*{oomen:06a}.}
\end{footnotesize}
\end{center}
\end{table}

Table \ref{table:descriptive-statistics} reports some descriptive statistics for our universe of stocks and sample period. As can be seen, these equities display varying degrees of liquidity with ORCL and SPY being the most liquid, while LMT and SLE are the least liquid. Also reported in the table is the univariate noise ratio statistic, $\gamma$, which is a noise-to-signal measure that describes the level of microstructure noise to integrated variance \citep*[see, e.g.,][for further details on the noise ratio]{oomen:06a}. Generally speaking, there is a tendency for more frequently traded companies to contain less microstructure noise, the notable exception being the transaction data for ORCL.

\subsection{Filtering procedures}
As a preliminary step, we subdued the sample data to some cleaning procedures. Pre-cleaning high-frequency data is necessary, because the raw data has many invalid observations (e.g., data with misplaced decimal points, or trades that are reported out-of-sequence). Our filter is roughly identical to that used by \citet*{barndorff-nielsen-hansen-lunde-shephard:08c} with some minor differences. Here, we briefly describe the filtering rules we employ.

\medskip

\noindent \textbf{Trades and quotes}: The following rules are applied to both trades and quotes data. a) We keep data from a single exchange: Pacific for SPY and primary exchange for the 4 remaining equities, see Table \ref{table:descriptive-statistics}, b) we delete data with time stamps outside the regular exchange opening hours from 9:30am to 4:00pm, c) we delete rows with a transaction price, bid or ask quote of zero, and d) we aggregate data with identical time stamp using volume-weighted average prices (using total transaction volume or quoted bid and ask volume, respectively).

\medskip

\noindent \textbf{Trades only}: We delete entries with a correction indicator $\neq$ 0 or with abnormal sales condition.

\medskip

\noindent \textbf{Quotes only}: We delete quotes with negative spreads and rows where the quoted spread exceeds 10 times the median spread for that day.

\medskip

\noindent Table \ref{table:descriptive-statistics} also reports how many observations that are lost by passing these filters through the data. It should be noted that the "Trades Only" and "Quotes Only" filters generally tend to reduce the sample by only a very small fraction.

\subsection{High-frequency covariance analysis}
Here, we inspect the outcome of applying the estimators introduced above, after which we look at transforms of the covariance matrix. As a comparison, we also compute the standard realised covariance from 15-min, 5-min, and 15-sec previous-tick data.

We implement both the $MRC \left[ Y \right]_{n}$ and $MRC \left[ Y \right]_{n}^{ \delta}$ estimators of section 2 and 3.4, respectively. Recall that $MRC \left[ Y \right]_{n}$ converges at rate $n^{-1 / 4}$, it needs to be corrected for bias, and as a result is not guaranteed to be positive semi-definite. We base $MRC \left[ Y \right]_{n}^{\delta}$ on $\delta = 0.1$, among many plausible choices, which results in a $n^{-1 / 5}$ rate of convergence, a small finite sample bias that we omit correcting and, hence, positive semi-definiteness by construction. Two sampling schemes are used, calendar time and refresh time sampling, which yields a total of four combinations. We use pre-averaging windows found as $k_{n}= \lfloor \theta n^{ \delta'} \rfloor$, where $\theta = 1$ and $\delta' = \left(0.5, 0.6 \right)$ for our two choices. The selection of $\theta$ follows the conservative rule discussed above. We set $n = 390$ for the calendar time-based estimators, which is 1-minute sampling, while $n$ is determined automatically by the data for the refresh time sampling scheme.

\begin{table}[H]
\setlength{\tabcolsep}{0.35cm}
\renewcommand{\arraystretch}{0.80}
\begin{center}
\caption{Average of high-frequency covariance matrix estimates}
\label{Table:avg-covar-mtx-estim}
\medskip
\begin{small}
\begin{tabular}{lrrrrrrrrrrrrr}
\hline
 & \multicolumn{5}{c}{$MRC[Y]_{n=CT(390)}$} &  & \multicolumn{5}{c}{$MRC[Y]_{n=CT(390)}^{\delta=0.1}$}\\
 & BMY & LMT & ORCL & SLE & SPY &  & BMY & LMT & ORCL & SLE & SPY\\
BMY &  1.465 &  0.203 &  0.259 &  0.154 &  0.233 &  &  1.396 &  0.195 &  0.259 &  0.152 &  0.231\\
LMT &  0.201 &  0.975 &  0.232 &  0.149 &  0.247 &  &  0.194 &  0.905 &  0.222 &  0.158 &  0.235\\
ORCL &  0.255 &  0.228 &  1.804 &  0.143 &  0.316 &  &  0.256 &  0.218 &  1.718 &  0.151 &  0.316\\
SLE &  0.144 &  0.149 &  0.136 &  0.955 &  0.162 &  &  0.144 &  0.156 &  0.143 &  0.911 &  0.167\\
SPY &  0.228 &  0.249 &  0.307 &  0.154 &  0.317 &  &  0.228 &  0.237 &  0.309 &  0.162 &  0.310\\
\\[-0.50cm]
 & \multicolumn{5}{c}{$MRC[Y]_{n=RT}$} &  & \multicolumn{5}{c}{$MRC[Y]_{n=RT}^{\delta=0.1}$}\\
 & BMY & LMT & ORCL & SLE & SPY &  & BMY & LMT & ORCL & SLE & SPY\\
BMY &  1.394 &  0.197 &  0.226 &  0.134 &  0.220 &  &  1.397 &  0.198 &  0.249 &  0.141 &  0.227\\
LMT &  0.193 &  0.955 &  0.224 &  0.134 &  0.242 &  &  0.201 &  0.924 &  0.222 &  0.147 &  0.241\\
ORCL &  0.196 &  0.198 &  1.756 &  0.128 &  0.297 &  &  0.228 &  0.220 &  1.726 &  0.139 &  0.310\\
SLE &  0.117 &  0.114 &  0.105 &  0.878 &  0.147 &  &  0.131 &  0.134 &  0.122 &  0.898 &  0.157\\
SPY &  0.200 &  0.237 &  0.243 &  0.120 &  0.310 &  &  0.216 &  0.249 &  0.285 &  0.138 &  0.311\\
\\[-0.50cm]
 & \multicolumn{5}{c}{$HY$} &  & \multicolumn{5}{c}{$HY[Y]_{n}^{(k,l)}$}\\
 & BMY & LMT & ORCL & SLE & SPY &  & BMY & LMT & ORCL & SLE & SPY\\
BMY &  1.160 &  0.093 &  0.097 &  0.103 &  0.111 &  &  1.417 &  0.192 &  0.255 &  0.144 &  0.227\\
LMT &  0.049 &  0.761 &  0.075 &  0.072 &  0.097 &  &  0.202 &  0.902 &  0.223 &  0.147 &  0.238\\
ORCL &  0.032 &  0.036 &  1.922 &  0.077 &  0.112 &  &  0.240 &  0.227 &  1.903 &  0.135 &  0.309\\
SLE &  0.029 &  0.026 &  0.013 &  1.053 &  0.095 &  &  0.134 &  0.137 &  0.127 &  0.900 &  0.159\\
SPY &  0.041 &  0.051 &  0.040 &  0.023 &  0.253 &  &  0.219 &  0.250 &  0.292 &  0.143 &  0.309\\
\\[-0.50cm]
 & \multicolumn{5}{c}{$Cov^{15s}$} &  & \multicolumn{5}{c}{$Cov^{avg(5m,15m)}$}\\
 & BMY & LMT & ORCL & SLE & SPY &  & BMY & LMT & ORCL & SLE & SPY\\
BMY &  1.823 &  0.088 &  0.118 &  0.091 &  0.120 &  &  1.499 &  0.192 &  0.288 &  0.149 &  0.221\\
LMT &  0.095 &  0.854 &  0.084 &  0.059 &  0.094 &  &  0.191 &  0.953 &  0.209 &  0.152 &  0.221\\
ORCL &  0.095 &  0.086 &  4.043 &  0.069 &  0.162 &  &  0.261 &  0.206 &  1.936 &  0.153 &  0.306\\
SLE &  0.056 &  0.049 &  0.043 &  1.947 &  0.072 &  &  0.146 &  0.148 &  0.142 &  1.008 &  0.161\\
SPY &  0.111 &  0.118 &  0.117 &  0.057 &  0.293 &  &  0.219 &  0.233 &  0.290 &  0.152 &  0.303\\
\hline
\end{tabular}\medskip
\end{small}
\begin{footnotesize}
\parbox{.95\textwidth}{\emph{Note}. This table reports average covariance matrix estimates. In all subpanels, the numbers in the upper diagonal (including diagonal elements) are based on transaction prices, while the lower diagonal is based on mid-quote data. $Cov^{avg(5m,15m)}$ is a simple time series average of the realised covariance computed from 5- and 15-minute returns.}
\end{footnotesize}
\end{center}
\end{table}

In Table \ref{Table:avg-covar-mtx-estim}, we report the sample covariance matrix estimates averaged across the 251 days. This table is constructed in the usual way, displaying the results based on transaction prices in the upper diagonal (including the main diagonal), while the strict lower diagonal elements are the corresponding results based on quotation data. Consistent with prior literature, we see from the table that the standard realised covariance is suffering from Epps effect, when sampling runs quickly. All estimated covariance terms lie in the positive region, but for $Cov^{15s}$ they are heavily compressed towards zero. This is less of a concern for $Cov^{avg(5m,15m)}$, which should tend to capture the average level of the covariance structure well, while not being seriously influenced by microstructure frictions and Epps effect. Turning next to the estimators proposed in this paper, we note that the time series average of both MRC versions and the noise-robust HY estimator are in line with that produced by $Cov^{avg(5m,15m)}$, showing that they appear free of any systematic bias. The \citet*{hayashi-yoshida:05a} estimator produces a strong downwards bias in the covariance estimates, when it is applied directly to noisy and irregular high-frequency data. This reaffirms previous empirical work \citep*[see, e.g.,][]{barndorff-nielsen-hansen-lunde-shephard:08c,griffin-oomen:06a,voev-lunde:07a}, so these results are not surprising or novel. The pre-averaged version $HY[Y]_{n}^{(k,l)}$, however, does a much better job and tends to agree with the average level of other noise-robust estimators.

Finally, we turn to the issue of positive semi-definiteness. As noted above, $MRC \left[ Y \right]_{n}$ and $HY[Y]_{n}^{(k,l)}$ are not guaranteed to possess this property. Nonetheless, they do not fail to be positive semi-definite on a single instance across our sample period. This is true for both the transaction and quotation data, and both combinations of the bias-corrected MRC estimator. Thus, while theoretically a concern, this problem does not appear to occur frequently in practice, although the conclusion might change for other data sets.

\subsection{Analysing realised beta}
We now focus on estimating $\beta^{(ji)}$ by $\hat{ \beta}_{n}^{(ji)}= MRC \left[ Y \right]_{n}^{i, j} / MRC \left[ Y \right]_{n}^{i, i}$, where we take $i = \text{SPY}$ and form regressions by using $j = \text{BMY, LMT, ORCL, SLE}$. This type of regression, where individual equity covariances with the market are regressed onto a market-wide realised variance measure, is important in financial economics, for example within the conditional CAPM \citep*[see, e.g.,][]{jagannathan-wang:96a,lettau-ludvigson:01a}, since only systematic risk should be rewarded with expected excess returns.

\begin{figure}[ht!]
\begin{center}
\caption{MRC-based beta.
\label{Figure:beta-mrc-vs-rc}}
\begin{tabular}{cc}
\footnotesize{Panel A: BMT vs. SPY} & \footnotesize{Panel B: LMT vs. SPY} \\
\includegraphics[height=6cm,width=0.45\textwidth]{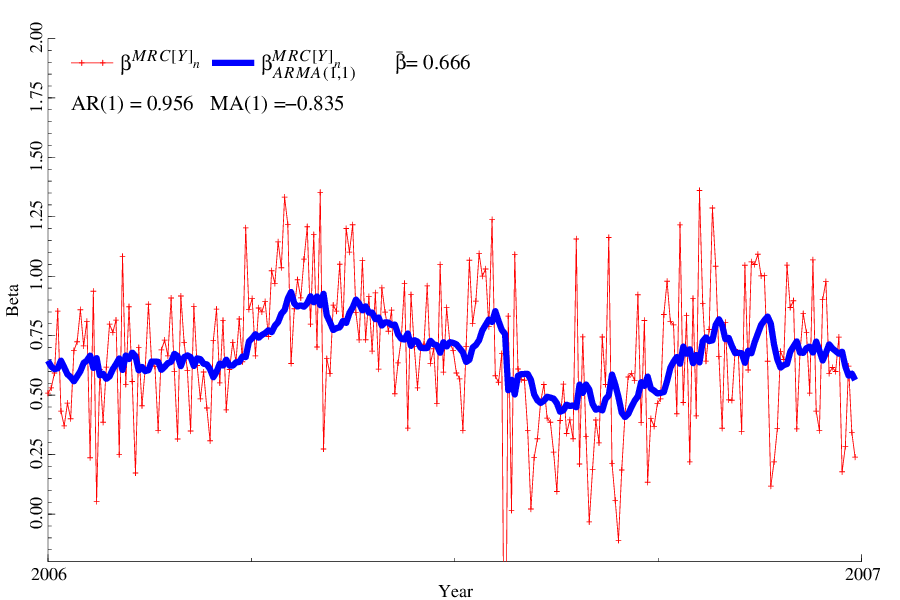} &
\includegraphics[height=6cm,width=0.45\textwidth]{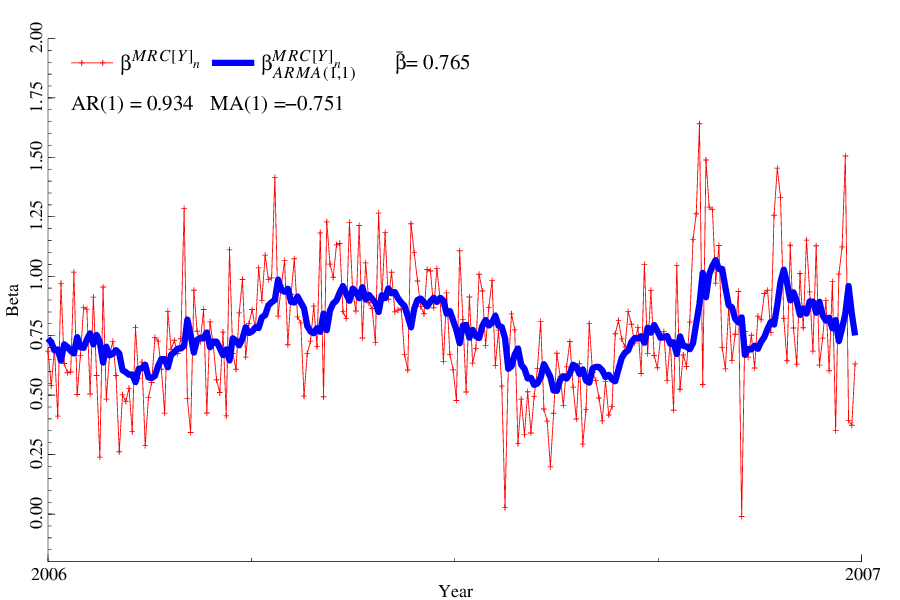} \\
\footnotesize{Panel C: ORCL vs. SPY} & \footnotesize{Panel D: SLE vs. SPY} \\
\includegraphics[height=6cm,width=0.45\textwidth]{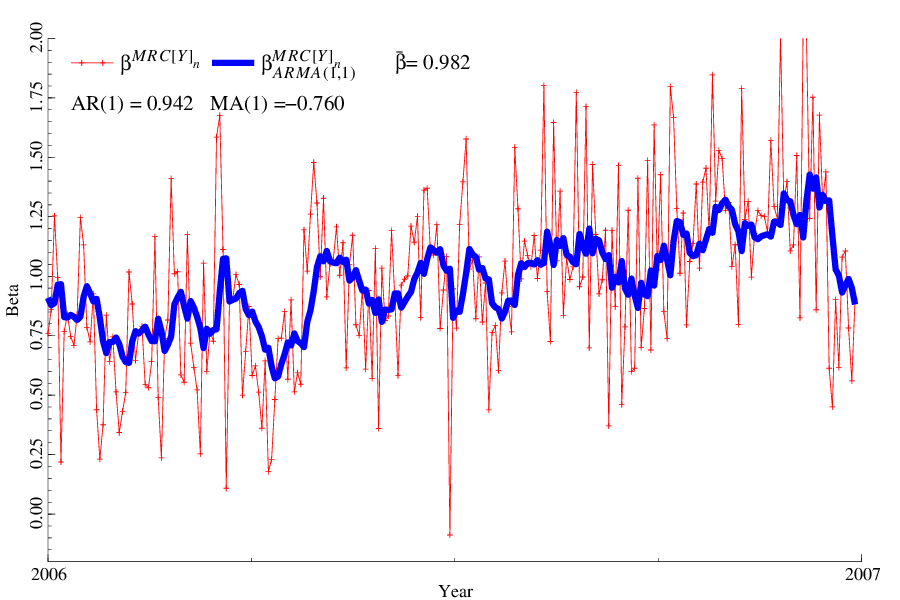} &
\includegraphics[height=6cm,width=0.45\textwidth]{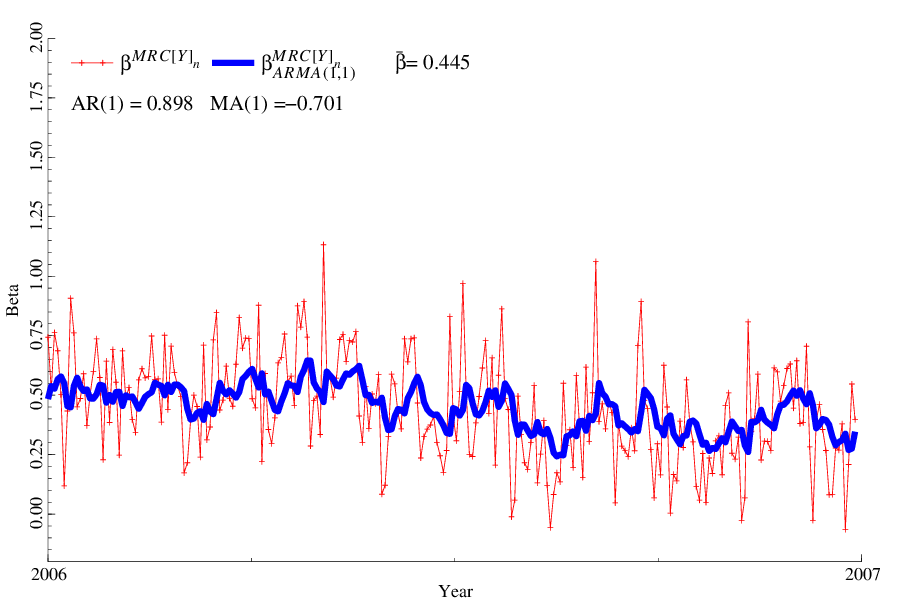} \\
\end{tabular}
\begin{footnotesize}
\parbox{0.94\textwidth}{\emph{Note}. $\beta^{MRC \left[ Y \right]_{n}}$ is the time series of daily $MRC \left[ Y \right]_{n}$-based beta estimates, using transaction prices and refresh time sampling for various asset combinations (in subpanels). $\beta_{ARMA(1,1)}^{MRC \left[ Y \right]_{n}}$ are fitted values from an ARMA(1,1) filter, with estimated autoregressive parameter AR(1) and moving average parameter MA(1). The sample mean MRC beta is reported as $\bar{ \beta}$.}
\end{footnotesize}
\end{center}
\end{figure}

In Figure \ref{Figure:beta-mrc-vs-rc}, we plot the MRC-based betas from transaction prices and refresh time sampling. The corresponding plots from the other estimators proposed in this paper are qualitatively similar. As in \citet*{barndorff-nielsen-hansen-lunde-shephard:08c}, we smooth the daily beta estimates by passing them through an ARMA(1,1) filter. The figure shows that beta is time-varying and predictable, and that it tends to fluctuate around its mean level. Evidently, the estimated processes exhibit substantial memory with autoregressive roots at 0.90 or higher  \citep*[see also][who study the persistence of quarterly realised beta estimated from daily asset returns]{andersen-bollerslev-diebold-wu:06a}.

\section{Concluding remarks}
In this paper, we present a simple solution, based on applying pre-averaging to financial high-frequency data, to the problem of how to estimate the multivariate ex-post integrated covariance matrix, possibly in the simultaneous presence of market microstructure noise and non-synchronous trading.

A modulated realised covariance (MRC) estimator is introduced. The MRC bears close resemblance to a standard realised covariance, being a sum of outer products of high-frequency returns, but it relies instead on pre-averaging to reduce the harmful impact of microstructure noise. We study the properties of this new estimator by showing its consistency and asymptotic mixed normality under mild conditions on the dynamics of the price process. As shown in the paper, the MRC can be configured to possess an optimal rate of convergence or to guarantee positive semi-definite covariance matrix estimates. In the presence of non-synchronous trading, we also outline how to modify the MRC by using an imputation scheme (for example the previous-tick rule or refresh time sampling) to match high-frequency prices in time. An MRC constructed on the back of such artificial returns will again be consistent for the integrated covariance.

Another novelty developed in this paper is a pre-averaged version of the \citet*{hayashi-yoshida:05a} estimator that can be implemented directly on the raw noisy \emph{and} non-synchronous observations, without any prior alignment of prices. We also show the consistency of this estimator and derive a rate for its variance, but otherwise we defer further theoretical analysis of finite sample improvements and asymptotic properties of the pre-averaged Hayashi-Yoshida estimator to future research \citep*[see][]{christensen-podolskij-vetter:10a}.

We demonstrate with a set of simulations that these newly proposed estimators can bring substantial efficiency gains with them compared to a standard realised covariance in the realistic setting with both microstructure noise and non-synchronous trading. Furthermore, an empirical illustration highlights their applicability to real high-frequency data. We therefore look forward to future applications of these estimators, including an investigation of their informational content about future volatility. Being able to produce good forecasts of future volatility is paramount in financial economics. Thus, as an example, it could be interesting to apply our estimators in the context of portfolio choice to calculate their economic value, akin to \citet*{fleming-kirby-ostdiek:01a,fleming-kirby-ostdiek:02a}, \citet*{bandi-russell:06a} and others. However, generally they should also be useful in many other areas, including asset- and option pricing or risk management.

\pagebreak

\appendix

\section{Proofs}

In the following, we assume that the processes $a$ and $\sigma$ are bounded. This is without loss of generality and can be justified by a standard localization procedure (see, e.g., \citet*{barndorff-nielsen-graversen-jacod-podolskij-shephard:06a}). Moreover, we denote constants by $C$, or $C_{p}$ if they depend on an additional parameter $p$. The main parts of the proofs are based upon \citet*{podolskij-vetter:09a} and \citet*{jacod-li-mykland-podolskij-vetter:09a}. \\[0.50cm]
\textit{Proof of Theorem \ref{consistency}:} Due to the triangular equality $[V, W]= \frac{1}{4} ([V + W, V + W] - [V - W, V - W])$, it suffices to prove the univariate case $d = 1$ (i.e. all processes are 1-dimensional). We use the decomposition
\begin{equation}
\label{decomp1} MRC \left[ Y \right]_{n} = \frac{1}{k_{n}} \sum_{l = 0}^{k_{n} - 1} MRC \left[ Y \right]_{n}^{l} - \frac{ \psi_{1}}{ \theta^{2} \psi_{2}} \hat{ \Psi}_{n},
\end{equation}
with
\begin{equation*}
MRC \left[ Y \right]_{n}^{l} = \frac{1}{ \theta \psi_{2} \sqrt{n}} \sum_{j = 0}^{[n / k_{n}] - 1} | \bar Y_{l + j k_{n}}^{n}|^2.
\end{equation*}
Notice that, for any $l = 0, \ldots, k_n - 1$, the summands in the definition of $MRC \left[ Y \right]_{n}^{l}$ are asymptotically uncorrelated. This type of estimators have been discussed in \citet*{podolskij-vetter:09a} and we can deduce by the methods presented therein (see the proof of Theorem 1) that
\begin{equation*}
MRC \left[ Y \right]_{n}^{l} \overset{p}{ \to} \int_0^1 \sigma _s^2 \text{d}s + \frac{ \psi_{1}}{ \psi_{2} \theta^{2}} \Psi,
\end{equation*}
where the convergence holds uniformly in $l$ (due to the boundedness of the processes $a$ and $\sigma$). On the other hand
we have that
\begin{equation*}
\hat{ \Psi}_{n} = \frac{1}{2n} \sum_{i = 1}^{n} | \Delta_{i}^{n} Y |^2 \overset{p}{ \to} \Psi.
\end{equation*}
This implies the convergence
\begin{equation*}
MRC \left[ Y \right]_{n} \overset{p}{ \to} \int_{0}^{1} \sigma_{s}^{2} \text{d}s,
\end{equation*}
which completes the proof. \qed \\[0.50cm]
\textit{Proof of Theorem \ref{clt2}:} Here we apply the "big blocks \& small blocks"-technique used in
\citet*{jacod-li-mykland-podolskij-vetter:09a}. The role of the small blocks (which will be asymptotically
negligible) is to ensure the asymptotic independence of the big blocks. More precisely, we choose an integer $p$, set
\begin{equation*}
a_i(p) = i(p+1)k_n \qquad \mbox{ and } \qquad b_i(p) = i(p+1)k_n + pk_n~,
\end{equation*}
and let
$A_i(p)$ denote the set of integers $l$ satisfying $a_i(p) \leq l < b_i(p)$ and $B_i(p)$ the
integers satisfying $b_i(p) \leq l < a_{i+1}(p)$. We further define $j_n(p)$ to be the largest integer $j$ such that $b_j(p) \leq n$ holds, which gives the identity
\begin{equation} \label{jnp}
j_n(p) = \Big\lfloor  \frac{n}{k_n (p+1)} \Big\rfloor -1.
\end{equation}
Moreover, we use the notation $i_n(p) = (j_n(p)+1)(p+1)k_n$.

Next, we introduce the random variable
\begin{equation} \label{approx}
\bar{Y}^{n}_{i, m} = \sum_{j = 1}^{k_n-1} g\Big( \frac{j}{k_n} \Big) (\sigma_{\frac mn} \Delta_{i+j}^n W + \Delta_{i+j}^n \epsilon)~,
\end{equation}
which can be interpreted as an approximation of some $\bar{Y}_{j}^{n}$. Moreover, we set
\begin{equation}
\Upsilon_{j,m}^n = \bar{Y}^{n}_{i, m} \left( \bar{Y}^{n}_{i, m} \right)' - \mathbb{E} \left[ \bar{Y}^{n}_{i, m} \left( \bar{Y}^{n}_{i, m} \right)'
|\mathcal F_{\frac{m}{n}} \right],
\end{equation}
and define
\begin{equation*}
\tilde{Y}_{j}^{n} =
\begin{cases}
\Upsilon_{j,a_i(p)}^n, &j \in A_i(p) \\
\Upsilon_{j,b_i(p)}^n, &j \in B_i(p) \\
\Upsilon_{j,i_n(p)}^n, &j \geq i_n(p)
\end{cases}
\end{equation*}
as well as
\begin{equation*}
\zeta(p,1)_j^n = \sum_{l=a_j(p)}^{b_j(p) -1} \tilde{Y}_l^{n}, \qquad \zeta(p,2)_j^n = \sum_{l=b_j(p)}^{a_{j+1}(p)-1} \tilde{Y}_l^{n}.
\end{equation*}
Notice that $\zeta(p,1)_j^n$ contains $pk_n$ summands ("big block") whereas $\zeta(p,2)_j^n$ contains $k_n$ summands ("small block").  Finally, we set
\begin{equation*}
\begin{array}{l}
M(p)^n = n^{-\frac 12} \sum_{j=0}^{j_n(p)} \zeta(p,1)_j^n, \qquad
N(p)^n = n^{-\frac 12} \sum_{j=0}^{j_n(p)} \zeta(p,2)_j^n, \qquad
C(p)^n = n^{-\frac 12} \sum_{j=i_n(p)}^{n} \tilde{Y}_j^{n}
\end{array}
\end{equation*}
and note that
\begin{equation}
\label{marting}
\mathbb{E} \left[ \zeta(p,1)_j^n | \mathcal F_{\frac{a_j(p)}{n}} \right] = 0 = \mathbb{E} \left[ \zeta(p,2)_j^n | \mathcal F_{\frac{b_j(p)}{n}} \right]
\end{equation}
by construction.

Now, by the same approximations as presented in \citet*{jacod-li-mykland-podolskij-vetter:09a} (see the identity (5.14), Lemma 5.5 and Lemma 5.6 therein) we get that
\begin{equation}
\label{ident1}
n^{1/4} \left( MRC \left[ Y \right]_{n} - \int_0^1 \Sigma_s \text{d}s \right) = \frac{n^{\frac{1}{4}}}{\theta \psi_{2}}(M(p)^n + N(p)^n +
C(p)^n) + R(p)_n
\end{equation}
where the last three summands satisfy the convergence
\begin{equation}
\label{ident2} \lim_{p \rightarrow \infty} \limsup_{n \rightarrow \infty} P( ||n^{ \frac{1}{4}} N(p)^n || +
||n^{ \frac{1}{4}} C(p)^n|| + ||R(p)^n ||> \delta ) = 0
\end{equation}
for any $\delta >0$. Notice that the term $R(p)_n$ stands for the approximation error in Eq. \eqref{approx}.

In the next lemma we show the stable convergence $\frac{n^{\frac{1}{4}}}{\theta \psi_{2}} M(p)^n \overset{d_{s}}{ \to} U(p)$ (for any fixed $p$). On
the other hand, we will see that, as $p\rightarrow \infty $, $U(p) \overset{p}{ \to} U$, where $U$ is the limiting variable defined in Theorem \ref{clt2}. By combining this with Eqs. (\ref{ident1}-\ref{ident2}) we obtain the assertion of Theorem \ref{clt2}.

\begin{lemma} \label{lemstab}
If the assumptions of Theorem \ref{clt2} are satisfied we obtain (for any fixed $p$)
\begin{equation*}
\frac{n^{ \frac{1}{4}}}{ \theta \psi_{2}} M(p)^n \overset{d_{s}}{ \to} U(p) = \sum_{j', k' = 1}^{d} \int_0^1 {\gamma_s^{jk,j'k'}(p) \text{\upshape{d}}B_s^{j'k'}},
\end{equation*}
and
\begin{eqnarray*}
&&\sum_{j,m=1}^{d}\gamma_s^{kl,jm} (p) \gamma_s^{k'l',jm} (p) =A_s^{kl,k'l'}= \frac{2}{\psi _2^2 }\left(  \frac{ \theta p}{p + 1}  \Lambda_s^{kl,k'l'}
\int_0^1 \Big(1 - \frac{u}{p} \Big)\phi_2^2(u) \text{\upshape{d}}u  \right. \\[1.5 ex]
&&\left. +
\frac{p}{ \theta(p+1)} \Theta_s^{kl,k'l'} \int_0^1 \Big(1 - \frac{u}{p} \Big)\phi_1(u) \phi_2(u) \text{\upshape{d}}u +
\frac{p}{ \theta^{3}(p + 1)}\Upsilon^{kl,k'l'} \int_0^1 \Big(1 - \frac{u}{p} \Big)\phi_1^2(u) \text{\upshape{d}}u \right),
\end{eqnarray*}
where the processes $\Lambda_s$, $\Theta_s$ and $\Upsilon$ are given in Theorem \ref{clt2}.
\end{lemma}
Notice that $\sum_{j,m=1}^{d}\gamma_s^{kl,jm} (p) \gamma_s^{k'l',jm} (p) \overset{p}{ \to} \sum_{j,m = 1}^{d}\gamma_s^{kl,jm}  \gamma_s^{k'l',jm} $ ($1\leq k,k',l,l'\leq d$), where $\gamma _s$ is defined
in Theorem \ref{clt2}. From this we deduce the convergence $U(p) \overset{p}{ \to} U$. \\ \\
\textit{Proof of Lemma \ref{lemstab}:} Due to Theorem IX 7.28 in \citet*{jacod-shiryaev:03a} the following conditions need to be shown (for all $1\leq k,k',l,l'\leq d$)
\begin{eqnarray}
\label{limit}
&& \frac{ n^{-1 / 2}}{ \theta^{2} \psi_{2}^{2}} \sum_{j = 0}^{j_n(p)} \mathbb{E} \left[ \zeta(p,1)_j^{n,kl} \zeta(p,1)_j^{n,k'l'}|\mathcal F_{\frac{a_j(p)}{n}} \right] \overset{p}{ \to} \int_0^1 A_u^{kl,k'l'} \text{d}u, \\[0.25cm]
\label{small} && n^{-1} \sum_{j=0}^{j_n(p)} \mathbb{E} \left[|| \zeta(p,1)_j^n||^4| \mathcal F_{\frac{a_j(p)}{n}} \right] \overset{p}{ \to} 0, \\[0.25cm]
\label{orth1} && n^{-1 / 4} \sum_{j=0}^{j_n(p)} \mathbb{E} \left[ \zeta(p,1)_j^{n,kl} \Delta W(p)_j^{n,k'} | \mathcal F_{ \frac{a_j(p)}{n}} \right] \overset{p}{ \to} 0, \\[0.25cm]
\label{orth} && n^{-1 / 4} \sum_{j=0}^{j_n(p)} \mathbb{E} \left[ \zeta(p,1)_j^{n,kl} \Delta N(p)_j^n |\mathcal F_{\frac{a_j(p)}{n}} \right] \overset{p}{ \to} 0,
\end{eqnarray}
where $\Delta V(p)_j^n = V_{n / b_j(p)} - V_{n / a_j(p)}$ for any process $V$ and Eq. \eqref{orth} holding for any 1-dimensional bounded martingale $N$ being orthogonal to $W$. For proving Eqs. \eqref{small} and \eqref{orth}, it is no restriction to assume that $d = 1$. Then these conditions
are already shown in \citet*{jacod-li-mykland-podolskij-vetter:09a} (Lemma 5.7). On the other hand, the functional $\zeta(p,1)_j^{n}$ is even in $W$. Since $W$ and $\epsilon$ are independent, we readily deduce that
\begin{equation*}
\mathbb{E} \left[ \zeta(p,1)_j^{n,kl} \Delta W(p)_j^{n,k'} \mid \mathcal F_{\frac{a_j(p)}{n}} \right] = 0,
\end{equation*}
which implies the condition in Eq. \eqref{orth1}. Hence, we are left to proving Eq. \eqref{limit}.

First, notice the identity
\begin{equation*}
\bar V_i^n  = \sum_{j = 1}^{k_{n}} g \left( \frac{j}{k_{n}} \right) \Delta_{i + j}^{n} V =
- \sum_{j = 0}^{k_n - 1} \left( g \left( \frac{j + 1}{k_{n}} \right) - g \left( \frac{j}{k_{n}} \right) \right) V_{ \frac{i + j}{n}}.
\end{equation*}
The second equality is useful for the computation of the moments of $\bar \epsilon_i^n$. By the smoothness assumption on the function
$g$ and the above identity we obtain the approximations ($1\leq k,l\leq d$)
\begin{equation} \label{expectation}
\mathbb{E} [ \overline{W}^{n,k}_{j} \overline{W}^{n,l}_{j'}]=\delta_{kl} \frac{k_{n}}{n} \psi_2 \bigg( \frac{|j-j'|}{k_{n}} \bigg) + O(n^{-1}), \quad
\mathbb{E} [ \overline{\epsilon }^{n,k}_{j} \overline{\epsilon }^{n,l}_{j'}]= \frac{ \Psi^{kl}}{k_{n}} \psi_1 \bigg( \frac{|j-j'|}{k_{n}} \bigg) + O(n^{-1})
\end{equation}
for $|j-j'|<k_n$, whereas the above expectations vanish when $|j-j'|\geq k_n$ (here $\delta_{kl}$ denotes the Kronecker symbol).
Next, we introduce the decomposition
\begin{equation*}
\zeta(p,1)_j^n=v(p,1)_j^n + v(p,2)_j^n +v(p,3)_j^n~,
\end{equation*}
where the terms $v(p,1)_j^n$, $v(p,2)_j^n$ and $v(p,3)_j^n$ are given by
\begin{eqnarray*}
v(p,1)_j^n &=& \sum_{l=a_j(p)}^{b_j(p) -1} \sigma_{ \frac{a_j(p)}{n}} \overline{W}^{n}_l
\Big( \sigma_{ \frac{a_j(p)}{n}} \overline{W}^{n}_l \Big)' -
\mathbb{E}  \left[\sigma_{ \frac{a_j(p)}{n}} \overline{W}^{n}_l
\Big(\sigma_{ \frac{a_j(p)}{n}} \overline{W}^{n}_l \Big)'|\mathcal F_{ \frac{a_j(p)}{n}} \right ], \\[0.25cm]
v(p,2)_j^n &=& \sum_{l=a_j(p)}^{b_j(p) -1} \overline{\epsilon}^{n}_l
\Big(\overline{\epsilon}^{n}_l \Big)' -
\mathbb{E}  \left[\overline{\epsilon}^{n}_l
\Big(\overline{\epsilon}^{n}_l \Big)' \right ], \\[0.25cm]
v(p,3)_j^n &=& \sum_{l=a_j(p)}^{b_j(p) -1} \sigma_{ \frac{a_j(p)}{n}}\overline{W}^{n}_l \Big(\overline{\epsilon}^{n}_l \Big)'+
\overline{\epsilon}^{n}_l \Big(\sigma_{ \frac{a_j(p)}{n}} \overline{W}^{n}_l \Big)'.
\end{eqnarray*}
By a straightforward calculation (and Eq. \eqref{expectation}) we obtain for all $1\leq k,l,k'l'\leq d$
\begin{eqnarray*}
\mathbb{E} [v(p,1)_j^{n,kl} v(p,1)_j^{n,k'l'}|\mathcal F_{ \frac{a_j(p)}{n}}] &=& \frac{2pk_n^4}{n^2}
\Lambda_{ \frac{a_j(p)}{n}}^{kl,k'l'} \int_0^1 \Big(1 - \frac{u}{p} \Big)\phi_2^2(u) \text{d}u + o_p(1), \\[0.25cm]
\mathbb{E} [v(p,2)_j^{n,kl} v(p,2)_j^{n,k'l'}|\mathcal F_{ \frac{a_j(p)}{n}}] &=&
2p \Upsilon^{kl,k'l'}  \int_0^1 \Big(1 - \frac{u}{p} \Big)\phi_1^2(u) \text{d}u + o_p(1), \\[0.25cm]
\mathbb{E} [v(p,3)_j^{n,kl} v(p,3)_j^{n,k'l'}|\mathcal F_{ \frac{a_j(p)}{n}}] &=&
\frac{2pk_n^2}{n} \Theta_{ \frac{a_j(p)}{n}}^{kl,k'l'} \int_0^1 \Big(1 - \frac{u}{p} \Big)\phi_1(u) \phi_2(u) \text{d}u + o_p(1),
\end{eqnarray*}
where the approximation holds uniformly in $j$. Now recall that $j_n(p) = \Big\lfloor  \frac{n}{k_n (p+1)} \Big\rfloor -1$.
Consequently, by Riemann integrability we deduce that
\begin{equation*}
\frac{n^{-\frac{1}{2}}}{\theta^2 \psi_2^2} \sum_{j=0}^{j_n(p)} \mathbb{E} [\zeta(p,1)_j^{n,kl} \zeta(p,1)_j^{n,k'l'} | \mathcal F_{ \frac{a_j(p)}{n}}]
\overset{p}{ \to} \int_0^1 A_u^{kl,k'l'} \text{d}u,
\end{equation*}
which completes the proof of Lemma \ref{lemstab}. \qed \\[0.50cm]

\textit{Proof of Theorem \ref{stochnon} and \ref{cltnon}:} Recall that $\frac{k_n}{n^{1/2 + \delta}} = \theta + o(n^{-1 / 4 + \delta / 2})$. A
straightforward calculation shows that
\begin{equation*}
\mathbb{E} \left[ MRC \left[ \, \epsilon \, \right]_{n}^{ \delta} \right] = \frac{ \psi_1}{ \theta^2 \psi_2 n^{2 \delta}} \Psi + o(n^{-1 / 4 + \delta / 2}).
\end{equation*}
By similar methods as presented in the proof of Theorem \ref{consistency} we deduce that
\begin{equation*}
MRC \left[ Y \right]_{n}^{ \delta} - \left( \int_{0}^{1} \Sigma_{s} \text{d}s + \frac{ \psi_{1}}{ \theta^{2} \psi_{2} n^{2 \delta}} \Psi \right) \overset{p}{ \to} 0.
\end{equation*}
Hence, we obtain the convergence in probability
\begin{equation*}
MRC \left[ Y \right]_{n}^{ \delta} \overset{p}{ \to} \int_{0}^{1} \Sigma_{s} \text{d}s,
\end{equation*}
which implies the assertion of Theorem \ref{stochnon}. Following the same scheme as demonstrated in the proof of Theorem \ref{clt2} we get
\begin{equation*}
n^{1/4 -\delta / 2} \left( MRC \left[ Y \right]_{n}^{ \delta} - \left( \int_{0}^{1} \Sigma_{s} \text{d}s + \frac{ \psi_{1}}{ \theta^{2} \psi_{2} n^{2 \delta}} \Psi \right) \right) \overset{d_{s}}{ \to} MN \left( 0, \frac{2 \Phi_{22} \theta}{ \psi_{2}^{2}} \int_{0}^{1} \Lambda_{s} \text{d}s \right),
\end{equation*}
since $\frac{ \psi_1}{ \theta^2 \psi_2 n^{2 \delta}} \Psi$ is an appropriate centering for the noise term in this case. Now
\begin{equation*}
\frac{ \psi_{1} n^{1 / 4 - \delta / 2}}{ \theta^{2} \psi_{2} n^{2 \delta}} \Psi  \overset{p}{ \to} 0
\end{equation*}
for $\delta >1/10$, whereas
\begin{equation*}
\frac{ \psi_{1} n^{1 / 4 - \delta / 2}}{ \theta^{2} \psi_{2} n^{2 \delta}} \Psi   = \frac{ \psi_{1}}{ \theta^{2} \psi_{2}} \Psi
\end{equation*}
for $\delta = 1 / 10$. Hence, Theorem \ref{cltnon} follows. \qed \\[0.50cm]
\textit{Proof of Theorem \ref{plimHYn}:} First, we start with the decomposition
\begin{eqnarray*}
HY[Y]_{n}^{(k,l)} &=& \frac{1}{ \left( \psi_{HY} k_{n} \right)^{2}}
\Big( \sum_{i,j} \bar{X^{k}}_{i}^{n} \bar{X^{l}}_{j}^{n} \mathbb{I}_{ \{ (t_{i}^{(k)}, t_{i+k_n}^{(k)}] \cap (t_{j}^{(l)}, t_{j+k_n}^{(l)}] \neq \emptyset \}} \nonumber \\[1.5 ex]
&+&\sum_{i,j}
(\bar{X^{k}}_i^n \bar{\epsilon^{l}}_j^n + \bar{\epsilon^{k}}_i^n \bar{X^{l}}_j^n)
\mathbb{I}_{ \{ (t_{i}^{(k)}, t_{i+k_n}^{(k)}] \cap (t_{j}^{(l)}, t_{j+k_n}^{(l)}] \neq \emptyset \}} \\[1.5 ex]
&+&\sum_{i,j}
\bar{\epsilon^{k}}_i^n \bar{\epsilon^{l}}_j^n
\mathbb{I}_{ \{ (t_{i}^{(k)}, t_{i+k_n}^{(k)}] \cap (t_{j}^{(l)}, t_{j+k_n}^{(l)}] \neq \emptyset \}} \Big) =: HY[Y]_{n}^1+ HY[Y]_{n}^2+ HY[Y]_{n}^3.
\end{eqnarray*}
As $X$ and $\epsilon$ are independent, it follows that $\mathbb{E}\left( HY[Y]_{n}^{2} \right) = 0$, and a simple computation shows that
\begin{equation*}
\text{var}(HY[Y]_{n}^2) \to 0.
\end{equation*}
Thus, $HY[Y]_{n}^2 \overset{p}{ \to} 0$. Next, we consider the term $HY[Y]_{n}^3$. This expression can be further decomposed as
\begin{eqnarray*}
HY[Y]_{n}^3 = \frac{1}{ \left( \psi_{HY} k_{n} \right)^{2}} \Big( \sum_{t_i\in J_{k,l}} a_{i}^n(k,l) \epsilon_{t_i}^k \epsilon_{t_i}^l + \sum_{t_i\in J_{k,l}^c} b_{i}^n(k,l) \epsilon_{t_i}^k \epsilon_{t_i}^l + \sum_{ (i,j)\in F_{k,l}} c_{ij}^n(k,l)
\epsilon_{t_i^{(k)}}^k \epsilon_{t_j^{(l)}}^l 1_{\{t_i^{(k)} \not= t_j^{(l)}\}}
\Big), \nonumber
\end{eqnarray*}
for certain numbers $a_{i}^n(k,l)$, $b_{i}^n(k,l)$ and $c_{ij}^n(k,l)$. Here $J_{k,l}$ denotes the set of common
points of $(t_i^{(k)})_{k_n\leq i\leq n_k-k_n}$ and $(t_i^{(l)})_{k_n\leq i\leq n_l-k_n}$, and
$J_{k,l}^c$ denotes the set of all common
points of $(t_i^{(k)})_{1\leq i\leq n_k}$ and $(t_i^{(l)})_{1\leq i\leq n_l}$ excluded $J_{k,l}$. The set $F_{k,l}$
is given by $$F_{k,l} =\{(i,j)|~\exists r,s ~\mbox{with}~r\leq i\leq r+k_n, s\leq j\leq s+k_n,
(t_{r}^{(k)}, t_{r+k_n}^{(k)}] \cap (t_{s}^{(l)}, t_{s+k_n}^{(l)}] \neq \emptyset \}.$$
Since $\bar{\epsilon^{k}}_i^n = -\sum_{j=1}^{k_n} \Big(g\Big(\frac{j}{k_n}\Big) - g\Big(\frac{j-1}{k_n}\Big) \Big)
\epsilon^{k}_{t_{i+j}^{(k)}}$ and $g$ is piecewise differentiable it holds that
\begin{equation*}
|b_{i}^n(k,l)| + |c_{ij}^n(k,l)| \leq C.
\end{equation*}
A straightforward computation shows that
\begin{equation*}
a_{i}^n(k,l) = \left( \sum_{j=1}^{k_n} g\Big(\frac{j}{k_n}\Big) - g\Big(\frac{j-1}{k_n}\Big) \right)^2= (g(1) - g(0))^2=0~,
\end{equation*}
because $g(0)=g(1)=0$. Thus, the first summand in the decomposition of $HY[Y]_{n}^3$ disappears, which is absolutely crucial for the proof. On the other hand, $\sharp J_{k, l}^{c} \leq C k_{n}$ which means that
\begin{equation*}
\frac{1}{ \left( \psi_{HY} k_{n} \right)^{2}} \sum_{t_i\in J_{k,l}^c} b_{i}^n(k,l) \epsilon_{t_i}^k \epsilon_{t_i}^l \overset{p}{ \to} 0.
\end{equation*}
Finally, note that the summands $\epsilon_{t_i^{(k)}}^k \epsilon_{t_j^{(l)}}^l \mathbb{I}_{\{t_i^{(k)} \not= t_j^{(l)}\}}$ have expectation $0$ and are mutually uncorrelated. Since $\sharp F_{k, l}\leq Cn k_n$ this implies that
\begin{equation*}
\frac{1}{ \left( \psi_{HY} k_{n} \right)^{2}} \sum_{(i,j)\in F_{k,l}} c_{ij}^n(k,l)
\epsilon_{t_i^{(k)}}^k \epsilon_{t_j^{(l)}}^l \mathbb{I}_{\{t_i^{(k)} \not= t_j^{(l)}\}} \overset{p}{ \to} 0.
\end{equation*}
Hence,
\begin{equation*}
HY[Y]_{n}^3 \overset{p}{ \to} 0.
\end{equation*}
Now we consider the term $HY[Y]_{n}^1$. We decompose
\begin{eqnarray}
HY[Y]_{n}^1 &=& \frac{1}{ \left( \psi_{HY} k_{n} \right)^{2}}
\Big( \sum_{(i,j)\in I_{k,l}} \bar{a}_{ij}^n(k,l)
\Delta_i^{n_k} X^{k} \Delta_j^{n_l} X^{l}
\mathbb{I}_{ \{ (t_{i-1}^{(k)}, t_{i}^{(k)}] \cap (t_{j-1}^{(l)}, t_{j}^{(l)}] \neq \emptyset \}} \nonumber \\[1.5 ex]
\label{decom} &+&\sum_{(i,j)\in I_{k,l}^c} \bar{b}_{ij}^n(k,l)
\Delta_i^{n_k} X^{k} \Delta_j^{n_l} X^{l}
\mathbb{I}_{ \{ (t_{i-1}^{(k)}, t_{i}^{(k)}] \cap (t_{j-1}^{(l)}, t_{j}^{(l)}] \neq \emptyset \}} \\[1.5 ex]
&+&\sum_{ (i,j)\in F_{k,l}} \bar{c}_{ij}^n(k,l)
\Delta_i^{n_k} X^{k} \Delta_j^{n_l} X^{l}
\mathbb{I}_{ \{ (t_{i-1}^{(k)}, t_{i}^{(k)}] \cap (t_{j-1}^{(l)}, t_{j}^{(l)}] = \emptyset \}}
\Big)~, \nonumber
\end{eqnarray}
for some constants $\bar{a}_{ij}^n(k,l), \bar{b}_{ij}^n(k,l), \bar{c}_{ij}^n(k,l)$, $I_{k,l}=\{ (i,j):~k_n\leq i\leq n_k-k_n,~
k_n\leq j\leq n_l-k_n\}$ and $I_{k,l}^c = \{ (i,j):~1\leq i\leq n_k,~
1\leq j\leq n_l\} - I_{k,l}$. Notice that all "border terms" are collected in the second summand whereas all terms with empty intersection of the intervals are in the third summand (in fact, we will see that both are negligible).

Notice that
\begin{equation*}
|a_{ij}^n(k,l)| + |b_{ij}^n(k,l)| + |c_{ij}^n(k,l)| \leq Cn.
\end{equation*}
Furthermore, we have
\begin{equation*}
\mathbb{E} \left[ |\Delta_i^{n_k} X^{k} \Delta_j^{n_l} X^{l} | \right] \leq C \sqrt{(t_{i}^{(k)} - t_{i-1}^{(k)})(t_{j}^{(l)} - t_{j-1}^{(l)})}
\end{equation*}
and $\sharp (I_{k,l}^c \cap \{(i,j):~(t_{i-1}^{(k)},
t_{i}^{(k)}] \cap (t_{j-1}^{(l)}, t_{j}^{(l)}] \neq \emptyset\}) \leq Ck_n$ by Eq. \eqref{schemereg}. This implies by Eq. \eqref{schemebound} that
\begin{equation*}
\frac{1}{ \left( \psi_{HY} k_{n} \right)^{2}} \sum_{(i,j)\in I_{k,l}^c} \bar{b}_{ij}^n(k,l)
\Delta_i^{n_k} Y^{k} \Delta_j^{n_l} Y^{l} \mathbb{I}_{ \{ (t_{i-1}^{(k)}, t_{i}^{(k)}] \cap (t_{j-1}^{(l)}, t_{j}^{(l)}] \neq \emptyset \}} = O_p(n^{-1/2})~,
\end{equation*}
and thus the second summand in Eq. \eqref{decom} is negligible.

Now, recall that $\Delta_i^{n_k} X^{k}$ can be replaced with $(\sigma_s\Delta_i^{n_k} W)^k$ for any $s$ with $t_i^{(k)} - s = O(k_n/n)$ ($1 \leq k \leq d$) without changing the first-order asymptotics \citep*[see, e.g.,][]{podolskij-vetter:09a}.
Notice also that the terms $\Delta_i^{n_k} W^{k} \Delta_j^{n_l} W^{l}
\mathbb{I}_{ \{ (t_{i-1}^{(k)}, t_{i}^{(k)}] \cap (t_{j-1}^{(l)}, t_{j}^{(l)}] = \emptyset \}}$ are mutually uncorrelated.
Hence,
\begin{equation*}
\frac{1}{ \left( \psi_{HY} k_{n} \right)^{2}}
\sum_{ (i,j)\in F_{k,l}} \bar{c}_{ij}^n(k,l) \Delta_i^{n_k} X^{k} \Delta_j^{n_l} X^{l}
\mathbb{I}_{ \{ (t_{i-1}^{(k)}, t_{i}^{(k)}] \cap (t_{j-1}^{(l)}, t_{j}^{(l)}] = \emptyset \}} = o_p(1).
\end{equation*}
Finally, consider the first summand in Eq. \eqref{decom}. As above, we approximate $\Delta_i^{n_k} X^{k} \Delta_j^{n_l} X^{l}$ by
\begin{equation*}
X_{i,j}^n (k,l) = (\sigma_{t_{i-1}^{(k)} \wedge t_{j-1}^{(l)}}\Delta_i^{n_k} W )^k
(\sigma_{t_{i-1}^{(k)} \wedge t_{j-1}^{(l)}}\Delta_j^{n_l} W )^l
\end{equation*}
whenever $(t_{i-1}^{(k)}, t_{i}^{(k)}] \cap (t_{j-1}^{(l)}, t_{j}^{(l)}] \neq \emptyset$, and set
\begin{equation*}
\overline{HY}[X]_{n}^{(k,l)} = \frac{1}{ \left( \psi_{HY} k_{n} \right)^{2}}
\sum_{(i,j)\in I_{k,l}} \bar{a}_{ij}^n(k,l) X_{i,j}^n (k,l)
\mathbb{I}_{ \{ (t_{i-1}^{(k)}, t_{i}^{(k)}] \cap (t_{j-1}^{(l)}, t_{j}^{(l)}] \neq \emptyset \}}.
\end{equation*}
Note that the quantities $HY[Y]_{n}^{(k,l)}$ and $\overline{HY}[X]_{n}^{(k,l)}$ have the same first order asymptotics. Then, it can be shown that
\begin{equation*}
\bar{a}_{ij}^n(k,l) = \left( \sum_{h=1}^{k_{n} - 1} g \left( \frac{h}{k_n} \right) \right)^2 = k_n^2 \left( \int_0^1 g(x) \text{d}x \right)^{2} + o \left( k_{n}^{2} \right).
\end{equation*}
Next, note that
\begin{equation}
\label{brcomp}
\mathbb{E} \left[ \Delta_i^{n_k} W^a \Delta_j^{n_l} W^b
\mathbb{I}_{ \{ (t_{i-1}^{(k)}, t_{i}^{(k)}] \cap (t_{j-1}^{(l)}, t_{j}^{(l)}] \neq \emptyset \}}| \mathcal{F}_{t_{i-1}^{(k)} \wedge t_{j-1}^{(l)}} \right]
= \delta_{ab} \left[\left( t_{i}^{(k)} \wedge t_{j}^{(l)} \right) - \left( t_{i-1}^{(k)} \vee t_{j-1}^{(l)} \right) \right].
\end{equation}
Eq. \eqref{brcomp} and Riemann integrability then delivers the convergence
\begin{equation*}
\frac{1}{ \left( \psi_{HY} k_{n} \right)^{2}}
\sum_{(i,j)\in I_{k,l}} \bar{a}_{ij}^n(k,l)
\mathbb{E} \left[ X_{i,j}^n (k,l)
\mathbb{I}_{ \{ (t_{i-1}^{(k)}, t_{i}^{(k)}] \cap (t_{j-1}^{(l)}, t_{j}^{(l)}] \neq \emptyset \}}|
\mathcal{F}_{t_{i-1}^{(k)} \wedge t_{j-1}^{(l)}} \right] \overset{p}{ \to} \int_{0}^{1} \Sigma_s^{kl} \text{d}s.
\end{equation*}
By usual martingale arguments we have that
\begin{equation*}
\overline{HY}[X]_{n}^{(k,l)} - \frac{1}{ \left( \psi_{HY} k_{n} \right)^{2}}
\sum_{(i,j)\in I_{k,l}} \bar{a}_{ij}^n(k,l)
\mathbb{E} \left[ X_{i,j}^n (k,l)
\mathbb{I}_{ \{ (t_{i-1}^{(k)}, t_{i}^{(k)}] \cap (t_{j-1}^{(l)}, t_{j}^{(l)}] \neq \emptyset \}}|
\mathcal{F}_{t_{i-1}^{(k)} \wedge t_{j-1}^{(l)}} \right] \overset{p}{ \to} 0.
\end{equation*}
On the other hand, $HY[Y]_{n}^{(k,l)} - \overline{HY}[X]_{n}^{(k,l)} \overset{p}{ \to} 0$.
Thus, collecting terms produces
\begin{equation*}
HY[Y]_{n}^{(k,l)} \overset{p}{ \to} \int_{0}^{1} \Sigma_s^{kl} \text{d}s,
\end{equation*}
which completes the proof. \qed \\[0.50cm]
\textit{Proof of Proposition \ref{orderTn}}: The terms
$\bar{Y^k}_i^n \bar{Y^l}_j^n \mathbb{I}_{\{ (t_i^{(k)}, t_{i+k_n}^{(k)}] \cap (t_j^{(l)}, t_{j+k_n}^{(l)}] \}}$
and $\bar{Y^k}_r^n \bar{Y^l}_s^n \mathbb{I}_{\{ (t_r^{(k)}, t_{r+k_n}^{(k)}] \cap (t_s^{(l)}, t_{s+k_n}^{(l)}] \}}$ are (asymptotically) uncorrelated when the intervals $(t_i^{(k)}, t_{i+k_n}^{(k)}]$ and $(t_j^{(l)}, t_{j+k_n}^{(l)}]$
do not intersect with $(t_r^{(k)}, t_{r+k_n}^{(k)}]$ and $(t_s^{(l)}, t_{s+k_n}^{(l)}]$. Thus, due to the assumption in Eq. \eqref{schemereg}, there are $O(n k_n^3)$ correlated terms, and each covariance has order $O(n^{-1})$ (cf. Eqn \eqref{Eqn:uorder}). This implies that $\text{var}(HY[Y]_n^{(k,l)})= O(n^{-1/2})$. \qed

\pagebreak

\ifx\undefined\BySame
\newcommand{\BySame}{\leavevmode\rule[.5ex]{3em}{.5pt}\ }
\fi \ifx\undefined\textmd
\newcommand{\textmd}[1]{{\sc #1}}
\fi \ifx\undefined\emph
\newcommand{\emph}[1]{{\em #1\/}}
\fi

\bibliographystyle{rfs}
\renewcommand{\baselinestretch}{1.0}

\end{document}